\documentclass[11pt,oneside]{amsart}

\usepackage{amsmath, amssymb, amsthm, mathtools}
\usepackage[dvipsnames]{xcolor}
\usepackage{hyperref}
\numberwithin{equation}{section}
\theoremstyle{plain}                
\newtheorem{theorem}{Theorem}[section]
\newtheorem{lemma}[theorem]{Lemma}
\newtheorem{proposition}[theorem]{Proposition} 
\newtheorem{corollary}[theorem]{Corollary}

\theoremstyle{definition}           
\newtheorem{definition}[theorem]{Definition}
\newtheorem{example}[theorem]{Example}

\newtheorem{assumption}[theorem]{Assumption}
\newtheorem{standingassumption}[theorem]{Standing Assumption}

\theoremstyle{remark}
\newtheorem{remark}[theorem]{Remark}

\newcommand{\p}[1]{\left( #1 \right)}
\newcommand{\bp}[1]{\big( #1 \big)}
\newcommand{\Bp}[1]{\Big( #1 \Big)}

 \newcommand{\tsM}{\tilde{\sM}}
 \newcommand{\tsQ}{\tilde{\sQ}}

\def\E{\mathbb{E}}
\def\P{\mathbb{P}}
\def\Q{\mathbb{Q}}

\DeclareMathOperator*\esssup{esssup}
\DeclareMathOperator*\essinf{essinf}

\newcommand{\tot}{\tfrac{1}{2}} 
\newcommand{\oo}[1]{\tfrac{1}{#1}}
\newcommand{\scl}[2]{\langle #1,#2 \rangle} 
\newcommand{\Bscl}[2]{\Big\langle #1,#2 \Big\rangle} 

\newcommand{\abs}[1]{\left| #1 \right|} 

\newcommand{\set}[1]{\{#1\}} 
\newcommand{\sets}[2]{\set{#1\,:\,#2}} 

\newcommand*{\QEDB}{\hfill\ensuremath{\square}}%

\newcommand{\bfone}{{\mathbf 1}}
\newcommand{\ind}[1]{ \bfone_{{#1}}} 
\newcommand{\inds}[1]{ \bfone_{\set{#1}}} 
\newcommand{\seq}[1]{\set{#1_n}_{n\in\N}} 

\newcommand{\prf}[1]{ \{ #1 \}_{t\in [0,T]}}

\newcommand{\dd}{d}

\newcommand{\tRN}[2]{\tfrac{\dd #1}{\dd #2}}

\newcommand{\downto}{\searrow}
\newcommand{\upto}{\nearrow}

\newcommand{\Implies}{\Rightarrow}

\providecommand{\R}{} \renewcommand{\R}{{\mathbb R}}

\providecommand{\N}{} \renewcommand{\N}{{\mathbb N}}

\newcommand{\PP}{{\mathbb P}}
\newcommand{\QQ}{{\mathbb Q}}

\newcommand{\EE}{{\mathbb E}}

\newcommand{\ee}[1]{ \bE \left[ #1 \right] }

\newcommand{\Bee}[1]{ \bE \Big[ #1 \Big] }

\newcommand{\MM}{{\mathcal M}}

\newcommand{\EN}{{\mathcal E}}

\newcommand{\ba}{{\mathrm{ba}}}

\newcommand{\eps}{\varepsilon}
\newcommand{\ld}{\lambda}

\newcommand{\gm}{\gamma}
\newcommand{\vp}{\varphi}

\newcommand{\el}{{\mathbb L}} 
\newcommand{\lzer}{\el^0}
\newcommand{\lone}{\el^1}

\newcommand{\linf}{\el^{\infty}}

\newcommand{\define}[1]{{\textbf{#1}}}

\newcommand{\efor}{\text{ for }}
\newcommand{\eforall}{\text{ for all }}

\newcommand{\eand}{\text{ and }}
\newcommand{\eor}{\text{ or }}
\newcommand{\ewhere}{\text{ where }}

\providecommand\st{{\mathcal t}}

\newcommand\tv{{\tilde{v}}}

\newcommand\sA{{\mathcal A}}

\newcommand\sC{{\mathcal C}}

\newcommand\sD{{\mathcal D}}

\newcommand\sE{{\mathcal E}}

\newcommand\bE{{\mathbb E}}

\newcommand\sF{{\mathcal F}}

\newcommand\sI{{\mathcal I}}

\newcommand\sK{{\mathcal K}}

\newcommand\sL{{\mathcal L}}

\newcommand\sM{{\mathcal M}}

\newcommand\sQ{{\mathcal Q}}

\newcommand\tS{{\tilde{S}}}

\newcommand\hX{{\hat{X}}}

\newcommand\hY{{\hat{Y}}}
\newcommand\sZ{{\mathcal Z}}

\renewcommand{\abs}[1]{|#1|}

\newcommand{\valu}{\mathfrak{U}}
\newcommand{\valv}{\mathfrak{V}}

\newcommand{\obv}{\mathbb{V}_B}
\newcommand{\obva}[1]{\mathbb{V}_{#1}}

\newcommand{\linfpp}{\linf_{++}}
\newcommand{\opt}{\hat{\sD}(B)}
\newcommand{\optz}{\hat{\sD}_0(B)}
\newcommand{\hmu}{\hat{\mu}}

\newcommand{\hpi}{\hat{\pi}}

\newcommand{\uB}{\underline{B}}

\newcommand{\sst}[1]{(#1 \cdot S)_T}
\renewcommand{\st}[1]{\left(#1 \cdot S\right)_T}

\newcommand{\hpis}{\hpi\cdot S}
\newcommand{\hpist}{(\hpis)_T}
\newcommand{\des}{\delta\cdot S}
\newcommand{\dees}{\delta^{\eps}\cdot S}
\newcommand{\dest}{(\des)_T}
\newcommand{\deest}{(\dees)_T}

\renewcommand{\oe}{\oo\eps}

\newcommand{\dele}{\delta_{\eps}}
\newcommand{\bev}{B+\eps\vp}
\newcommand{\ubev}{\underline{\bev}}
\newcommand{\pie}{\pi_{\eps}}

\begin{document}
\title{Conditional Davis pricing}

\author{Kasper Larsen}
\address{Kasper Larsen, Department of Mathematics, Rutgers University}
\email{kasperl@andrew.cmu.edu}

\thanks{The authors would like to thank Pietro Siorpaes, Mihai S\^\i
rbu, and Kim Weston for numerous and helpful discussions. During the
preparation of this work the first author has been supported by the
National Science Foundation under Grant No.~DMS-1411809 (2014 - 2017)
and Grant No.~DMS-1812679  (2018 - 2021), the second author has been
supported by the Swiss National Foundation through the grant SNF
$200021\_153555$ and by the Swiss Finance Institute, and the third
author   has been supported by the National Science Foundation under
Grant No.~DMS-1107465 (2012 - 2017) and   Grant No.~DMS-1516165 (2015 -
2018).  Any   opinions, findings and conclusions or recommendations
expressed in this   material are those of the author(s) and do not
necessarily reflect the   views of the National Science Foundation
(NSF)}

\author{Halil Mete Soner}
\address{Halil Mete Soner, Department of Mathematics, ETH Z\" urich}
\email{mete.soner@math.ethz.ch}

\author{Gordan \v{Z}itkovi\'{c}}
\address{Gordan \v Zitkovi\' c, Department of Mathematics,
University of Texas at Austin}
\email{gordanz@math.utexas.edu}

\subjclass[2010]{Primary 91G10, 91G80; Secondary 60K35.
\\\indent\emph{Journal of Economic Literature (JEL)
Classification:} C61, G11}

\keywords{Incomplete markets, utility-maximization, unspanned
endowment, local martingales,  linearization, directional derivative.}


\begin{abstract}
We study the set of marginal utility-based prices of a financial
derivative in the case where the investor has a non-replicable random
endowment. We provide an example showing that even in the simplest of
settings - such as Samuelson's geometric Brownian motion model - the
interval of marginal utility-based prices can be a non-trivial
strict subinterval of the set of all no-arbitrage prices. This is in
stark contrast to the case with a replicable endowment where non-
uniqueness is exceptional. We provide formulas for the end points for
these prices and illustrate the theory with several examples.
\end{abstract}

\maketitle

\begin{center}
\today
\end{center}\ \\

\section{Introduction}
\label{sec:intro}

We consider an investor in a frictionless but incomplete financial
market.  The  price dynamics are modeled by a locally bounded
semimartingale $S$. The investor will receive an endowment $B$ at a
future time $T>0$ and would like to price a derivative with
payoff $\varphi$ at time $T$. To obtain the set of possible prices for the payoff 
$\varphi$, we use the  marginal utility-based pricing approach of Mark
Davis in \cite{Dav97}. With the investor's utility function
$U:(0,\infty) \to \R$ given, we start by defining the primal value
function $v(\cdot;B)$ ``conditioned'' on the presence of the
endowment $B\in \mathbb{L}^\infty_{++}:=
\cup_{a>0}(a+\mathbb{L}^\infty_{+})$. Its domain is a set of random
variables $\vp$ interpreted as future derivative payments, and its value
is
\begin{align}
\label{intro1}
v(\vp;B) := \sup_{H} \E[ U(\varphi+ B + \int_0^T H_u dS_u)],\quad
 \varphi \in \mathbb{L}^\infty,
\end{align}
where $H$ ranges over a set of admissible integrands which is defined in
Section \ref{sec:setup} below. Then, given a fixed derivative with payoff 
$\vp$, we call a constant $p \in \R$ a  \emph{conditional Davis price
of $\varphi$} (conditional  on the endowment $B$), if $p$ satisfies the
following inequality
\begin{align}
\label{intro2}
v\big(\eps (\varphi-p); B\big)\le v(0;B) ,\quad \forall \eps \in \R.
\end{align}
The mathematical details are given in Definition
\ref{def:Davis} below.

\medskip

The conditional Davis pricing concept above can also be seen
as a variation of the classical case where the utility function is no
longer deterministic. We could consider the random endowment $B$
a part of the preference structure of the agent, i.e., think of
$x\mapsto U(x+B(\omega))$ as a stochastic utility function and view
$\EE[ U(\vp+B)]$ as the expected utility of the position $\vp$. 

\medskip

The set of spanned endowments $B$ provides an interesting
special case. Indeed, if $B\in \linf_{++}$ satisfies $B = x+
\int_0^T H_u dS_u$ for some initial value $x$ and an admissible $H$,
then $x$ is unique and the primal utility function satisfies $v(\vp;B)
=v(\vp;x)$.  In this case, the conditional Davis price is exactly the
marginal utility-based price of the payoff $\varphi$ given the (constant)
initial  wealth $x>0$  defined in Definition 3.1 in \cite{HugKraSch05}.
This paper's main goal is to extend the theory developed for spanned
endowments to the case of unspanned endowments, which is one of
the main problems one faces in many incomplete-market-equilibrium
frameworks.

\medskip

An alternate approach is to consider the endowment $B$ as the payoff of
a financial derivative as well.  This perspective requires us
to price multiple derivatives simultaneously and was considered
in \cite{HugKra04,KS06}. Namely, let $\xi$ be the $d$-dimensional
payoff of $d$ derivatives and fix an amount $q_0 \in \R^d$.  Then, in
Remark 1 of \cite{HugKra04}, $p_0 \in \R^d$ is called the \emph{marginal
utility-based price of $\xi$ given $q_0$} if $p_0$ satisfies, $$ v(
q\cdot (\xi-p_0);q_0 \cdot \xi) \le v(0;q_0 \cdot \xi), \quad \forall q
\in \R^d. $$ So alternatively, one may define  conditional Davis prices
of $\varphi$ as the second component of the two-dimensional marginal
utility-based prices of the pair $\xi:=(B,\varphi)$ with $q_0:=(1,0)$.
Indeed, we give the precise definition in Section \ref{s.new} and in
Section \ref{ss.i} we show that these two notions agree.

\medskip

Apart from using the same definition, previous studies do not cover the
case of conditional Davis prices with an unspanned endowment $B$. In
particular, \cite{HugKraSch05} study the one dimensional problem with a
constant (spanned) endowment $B:=x>0$, while \cite{HugKra04} provide the
multi-dimensional definition, but only investigate the utility
maximization problem. \cite{KS06} study the multi-dimensional pricing
problem only for small values of $q_0$ under a decay assumption on $\xi$
which we discuss later. Since we need to fix $q_0$ to be $(1,0)$, the
asymptotic results of \cite{KS06} cannot be applied to the general
conditional Davis prices. To highlight one non-trivial
difference between our setting and \cite{KS06}, let us recall that it
has been known since \cite{HugKraSch05} that even when $B:=x>0$ is
constant, marginal utility-based prices of the payoff $\vp$ can form a non-trivial
interval. Such an occurrence is, however,  treated as a rare
pathology and explicitly assumed away via a decay condition in
\cite{KS06}. In our setting, where $B$ is unspanned, we provide an
example showing that even in Samuelson's geometric Brownian motion model
with constant coefficients, there exists a whole spectrum of
explicit bounded payoffs $\varphi$ with a non-trivial and
explicitly computable interval of marginal utility-based prices.

\medskip

While the notion of marginal utility-based prices has been around for
more than two decades (we discuss related literature below), a
characterization of all constants $p$ satisfying \eqref{intro2} when
both $B$ and $\varphi$ are unspanned is currently not available. The
main contribution of this paper provides a set of conditions imposed on
$(B,\vp)$ under which the two endpoints of the interval of marginal
utility-based prices for the payoff $\vp$ can be explicitly computed.
Our main results are:

\begin{enumerate}

\item For an arbitrary endowment $B\in\mathbb{L}^\infty_{++}$, we show
that the interval of marginal utility-based prices for the payoff
$\varphi\in \mathbb{L}^\infty(\PP)$ is given by the set of all
$\scl{\vp}{\hat{\QQ}}\subset\R$, when $\hat{\Q}\in$ba$(\P)$ ranges
through the set of finitely-additive minimizers of the associated dual
utility problem (as introduced in \cite{CSW01}). Here $\langle
\cdot,\cdot\rangle$ denotes the dual pairing between
$\mathbb{L}^\infty(\PP)$ and its dual ba$(\P):=
(\mathbb{L}^\infty(\P))'$.

\item It is well-known that marginal utility-based prices are linked to
derivatives of the primal utility value function \eqref{intro1}. We show
that under a mild growth condition on the utility function $U(\xi)$ at
$\xi=0$, the directional derivative of $v(\vp;B)$ in the direction of
$\vp$ can be characterized as the value of a certain linear stochastic
control problem. Furthermore, we show by means of an example that even
in the $log$-utility case, the mapping

\begin{align}
\nonumber
\R \ni x \to v(x;B);
\end{align}
that is, the restriction of $v(\cdot;B)$ to constant payoffs
$\vp=x\in\R$ 
can fail to be differentiable on the interior of its
domain\footnote{When $B:=x>0$ is constant, Theorem 2.2 in \cite{KS99}
ensures smoothness of the primal value function \eqref{intro1}. However,
when $B\in \mathbb{L}^\infty_{++}$ is unspanned, Example \ref{CSW} below
illustrates that smoothness can fail.}. See also the discussion in
Erratum \cite{CSW17}.

\item Under the additional assumption of unique super-replicability
from \cite{LSZ16} placed on $(B,\vp)$, we solve the linear stochastic
control problem mentioned in (2) explicitly. This gives us formulas for
the two endpoints describing the interval of marginal utility-based
prices. As an offshoot, we show additionally that the mapping $ \varphi
\to v(\varphi;B)$ is smooth whenever $B$ is uniquely super-replicable.
\end{enumerate}

We wish to stress that while some known results related to marginal
utility-based prices extend verbatim from the constant endowment case
$B:=x>0$ to the general, unspanned case of $B\in\linf_+$, not all
results do. Indeed, the example mentioned in (2) above illustrates a
non-trivial difference.

\medskip

Because of market incompleteness, the interval of arbitrage-free prices
for $\varphi$ often takes on the extreme form with endpoints given by
the essential suprema and infima of $\vp$. We refer the reader to
monographs \cite{Car09} and \cite{FS04} for further general information
and thorough historical overviews of various ways to price unspanned
payoffs. Several authors have used ad hoc methods to reduce the width of
the interval of arbitrage-free prices (see, e.g., \cite{CSR00} and its
extension \cite{BS06} where so-called good deal bounds, based on the
Hansen-Jagannathan bound for the Sharpe ratio, are used). In this
respect, our results show how marginal utility-based prices can be used
to narrow down the interval of arbitrage-free prices.

\medskip

We continue by elaborating on the lack of uniqueness of marginal
utility-based prices mentioned above. For utility functions
defined on the positive axis (such as power and log), and with $B:=x>0$
constant, \cite{HugKraSch05} show that Davis prices are unique for all
$\varphi\in \mathbb{L}^\infty$ if and only if the dual utility optimizer
$\hat{\QQ}$ is a martingale measure (in which case, the unique marginal
utility-based price of the payoff $\varphi$ is given by
$\E^{\hat{\QQ}}[\varphi]$). However, it has been known since
\cite{KLSX91} that there exist arbitrage-free models where the dual
utility optimizers fail the martingale property; see also \cite{KS99}
for further examples of models satisfying NFLVR  (no free lunch with
vanishing risk) where the dual utility optimizers fail the martingale
property. As a consequence, there are models for which there exists a
bounded payoff $\varphi$ with non-unique marginal utility-based prices
(see \cite{HugKraSch05} for an abstract construction of such a payoff
$\varphi$). When uniqueness is considered indispensable, one
could restrict
attention to financial models and utility functions which produce
martingale dual utility optimizers (see, e.g., the BMO-type condition
used in \cite{KW16}), or consider only payoffs $\varphi$ with unique
Davis prices (as done in \cite{KS06}). We do not impose such
restrictions and our intervals of marginal
utility-based prices are generally nontrivial.

We finish this introduction with  a brief summary of the sizable
literature on  marginal utility-based prices in the case when $B:=x>0$
is constant  (more generally, when $B\in \mathbb{L}^\infty_{++}$ is
spanned).  The strand of literature that comes closest to this paper
where $U$ is finite only on the positive axis includes \cite{CSW01},
\cite{HugKra04}, \cite{KS06}, and \cite{HugKraSch05}.\footnote{When
$U$'s domain is $\R$ (such as exponential utility) the corresponding
dual optimizer is always a martingale; see \cite{BF02}. Consequently,
for such utility functions, marginal utility-based prices are always
unique and our analysis offers nothing new in that case.}  Let us
comment on their similarities and differences with our paper:

\indent - \cite{CSW01} focus on the utility-maximization problem itself
and do not consider pricing.

\indent - The notion of marginal utility-based prices is defined in
Remark 1 on page 849 of \cite{HugKra04}, and is not studied beyond a
standard super-differential characterization.

\indent - The authors of \cite{HugKraSch05} perform an in-depth study of
marginal utility-based prices in the constant endowment case, i.e.,
$B:=x>0$ for some constant $x$ (or more generally for $B$ spanned).
\cite{HugKraSch05} created the first abstract example exhibiting non-
unique marginal utility-based prices.

\indent - In \cite{KS06}, the authors perform an asymptotic analysis. As
discussed earlier and further detailed in Remark \ref{rmk:newMete}
below, these results do not apply to our setting. Moreover, \cite{KS06}
work under assumptions guaranteeing that $\vp$ has a unique marginal
utility-based price. These assumptions come in the form of a decay
condition on $\vp$ and can be found  already in \cite{HugKraSch05}.
This decay condition, in particular, is not satisfied for a
generic bounded payoff. We do not require such a decay property and
consequently, the set of claims we consider and the set considered in
\cite{KS06} do not nest in either direction. Indeed, based on the famous
counterexample in \cite{DS98b}, we construct  an explicit example of a
family of payoffs with constant $B:=x>0$ (as in \cite{KS06}) which has a
non-trivial interval of  marginal utility-based prices.   This example
illustrates that even for a constant endowment $B:=x>0$, our   setting
allows for non-unique prices whereas the setting of   \cite{KS06} always
produces a unique marginal-utility based price.

\medskip

The paper is organized as follows. The model is described, the
terminology and notation set, standing assumptions imposed,  and
preliminary analysis of our central utility-maximization problem is
performed in Section \ref{sec:setup}. In Section
\ref{sec:davis} we define conditional Davis prices. In Section
\ref{s.new} we recall the definition of marginal utility-based prices
from  \cite{HugKra04} and we show that conditional Davis prices can be
seen as a projection of marginal-utility based prices. Section
\ref{sec:characterization} characterizes the Davis prices from the dual
point of view and lays out some of the first consequences of this
characterization. Directional derivatives of the primal utility-
maximization problem are studied in Section \ref{sec:directional} which
also gives the explicit example of non-smoothness mentioned in (2)
above.  Section \ref{sec:directional} also gives a characterization of
the directional derivative in terms of a linear stochastic control
problem. Section \ref{sec:unique} recalls the definition of unique
super-replicability from \cite{LSZ16} and provides a family of examples
of uniquely super-replicable claims. The main result of Section
\ref{sec:unique} gives an explicit expression for the directional
derivative of the utility-maximization value function under the unique
super-replicability condition. This result is subsequently used in
Section \ref{sec:interval} to give explicit formulas for the interval of
marginal utility-based prices in a general setting. These formulas are
then used in two examples where one example is set in a Samuelson-Black-
Scholes-Merton type model and illustrates the fact that non-uniqueness
can arise even in the simplest of settings (this supports our claim in
the abstract). These examples also allow us to give explicit expressions
for the first-order approximation of the hedging portfolios associated
to the end-points of the interval of marginal utility-based prices.

\section{The setup and assumptions }
\label{sec:setup}
\subsection{The market model}
Let $(\Omega,\sF,\prf{\sF_t},\PP)$ be a filtered probability space which
satisfies the usual conditions, and let $\prf{S_t}$ be a locally bounded
semimartingale. $L(S)$ denotes the set of all predictable $S$-integrable processes
and $\sM$ denotes the set of all $\PP$-equivalent countably-additive probability
measures $\QQ$ on $\sF$ for which $S$ is a $\QQ$-local martingale.
\begin{standingassumption}[NFLVR]
\label{ass:NFLVR}
$\sM\ne \emptyset$.
\QEDB
\end{standingassumption}
\begin{remark}\label{rem:NFVLR}
We assume that the asset-price process $S$ is locally bounded and postulate
the existence of a local martingale measure. While it is
possible to relax our setting to the non-locally-bounded case (as used in,
e.g., \cite{CSW01}), it is not be possible to relax Assumption
\ref{ass:NFLVR} so as to imply the existence of a supermartingale deflator
only. Indeed,
the presence of a non-replicable endowment $B$ makes the admissibility
class which produces only nonnegative wealth processes too small to host an optimizer.
This delicate issue is discussed and illustrated on
pages 240, 241 in \cite{Lar09}. To keep the focus of the current paper on the
issues directly related to conditional Davis pricing, we have
opted for a set of assumptions which is slightly stronger than
absolutely necessary.
\end{remark}

\subsection{Gains and admissibility}
The investor's gains process has the following dynamics
\begin{align}\label{dX}
  (\pi\cdot S)_t := \int_0^t \pi_u\, dS_u,\quad t\in [0,T],
 \end{align}
 for some $\pi \in L(S)$.
 We call $\pi\in L(S)$ admissible if the gains process is uniformly lower bounded by a constant in which case we write $\pi \in \sA$.
The set
of terminal outcomes is denoted by $\sK$, i.e., we define
\[ \sK := \sets{ (\pi\cdot S)_T}{ \pi \in \sA}.\]
\subsection{The primal problem}
Let $U$ be a utility function on $(0,\infty)$, i.e.,
$U$ is strictly concave, strictly increasing, and continuously differentiable with $U'(0+)=+\infty$ and $U'(+\infty)=0$.
When necessary, we extend the domain of $U$ to $\R$ by setting $U(x) = -\infty$
for $x<0$ and $U(0) = \inf_{x>0} U(x)$.
Finally, $U$ is said to be reasonably
  elastic (as defined in \cite{KS99}) if
  \[ \limsup_{x\to\infty} \tfrac{ xU'(x)}{U(x)} < 1.\]
Even though we need it for some of our results, we do not impose the condition of
reasonable elasticity from the start.

\medskip

Let  $v$ be defined
as in \eqref{intro1} with
the above notion of admissibility.
Then, for any $B\in\linfpp:=\cup_{x>0} (x+\linf_+)$ we define
\begin{align}\label{equ:u}
\valu(B) &:= v(0;B)=
\sup_{X \in \sK}\, \Bee{ U\Big(B+X
\Big) }
 \end{align}
 with the convention that
 $\EE[  U(B+X)]=-\infty$ if $\EE[  U(B+X)^-]=+\infty$.
Because $\valu(B) \geq U(\essinf B)>-\infty$,
$\valu$ is $(-\infty,\infty]$-valued on
$\linfpp$. In \eqref{ass:proper} below, we impose
a dual properness assumption which among other things ensures that $\valu$ is finitely valued on
$\linfpp$.

\subsection{The dual utility maximization problem}

The set of equivalent local martingale measures $\sM$ can be identified - via
Radon-Nikodym derivatives with respect to $\PP$ - with a subset
of $\mathbb{L}_+^1(\PP)$ and embedded, naturally, into
$\ba(\PP):=\linf(\PP)^{*} \supseteq \lone(\PP)$.  We define
$\overline{\MM}^*$ as the weak$^*$-closure of $\sM$ and we define
$\sD\subset \ba_+(\PP)$ as the family of all $y \QQ$ where  $y\in[ 0,\infty)$ and $\QQ\in\overline{\MM}^*$.
We can then define the dual utility functional  by

\[ \obv(\mu) := \sup_{X\in\linf} \Big( \E[U(B+X)] - \scl{\mu}{X} \Big),\quad \mu\in\ba(\PP).\]
In particular,  $\obv$
is convex,
lower weak$^*$-semi\-continuous on
$\ba(\PP)$
and bound\-ed from below by $\E[U(B)]\in\R$. For the reminder of the paper we impose a properness assumption.
While not the weakest possible in our setting, this assumption allows
us to deal swiftly, and yet with a minimal loss of generality,
with several technical points that are not central to
the message of the paper:
\begin{standingassumption}[Properness]  There exist $y_0\in(0,\infty)$ and $\QQ_0\in\MM$ such that $\mu_0 := y_0\QQ_0$ satisfies
\label{ass:proper}
 \begin{equation}
 \label{equ:mu0}
 \begin{split}
    \obv(\mu_0)<\infty.
 \end{split}
 \end{equation}
   \QEDB

  \end{standingassumption}

Thanks to a minimal modification of Lemma 2.1 on p.~138 in \cite{OZ09} and the discussion before it,  $\obv$ admits the following representation
   \begin{align}
   \label{equ:rep}
   \obv(\mu) = \ee{ V\left(\tRN{\mu^r}{\PP}\right)} +
   \scl{\mu}{B},\quad \mu\in\sD,
   \end{align}
where $V$ is the dual utility function (strictly convex) defined by
$$
 V(y) := \sup_{x>0} \Big( U(x) - xy \Big),\quad y>0.
$$
Consequently, Fenchel's inequality and \eqref{equ:mu0} guarantee that the primal value function $\valu$ satisfies $\valu(B)<\infty$ for all $B\in\linfpp$. Furthermore, \eqref{equ:mu0} also ensures that
 the corresponding dual value function defined by
\begin{align}
\label{equ:dual}
   \valv(B) := \inf_{\mu\in\sD} \obv(\mu),\quad B\in\linfpp,
\end{align}
is finitely valued. For $B\in\linfpp$ we let $\opt$ denote the set of all minimizers, i.e., all those $\mu\in\sD$ such that $\valv(B) = \obv(\mu)$.

 The next result collects some basic facts we will need in the following sections:

\begin{lemma}\label{lem:minimizer}
For each $B\in\linfpp$, the set $\opt$ is a nonempty weak$^*$-compact subset
of $\ba(\P)$ and there exists
a nonnegative random variable $\hY=\hY(B)$ such that  $\PP[\hY>0]>0$ and
\[  \hY = \tRN{\mu^r}{\PP} \quad  \text{
for all $\mu\in\opt$.}\]
{Furthermore, the strong duality $\valu(B) = \valv(B)$ holds for all $B\in\linfpp$.}

\end{lemma}
\begin{proof} {
Let $\seq{\mu}$ be a minimizing sequence for $\obv$ of the form
\[ \mu_n = y_n \QQ_n \text{
where $\QQ_n\in \overline{\sM}^*$ and $\seq{y}\subseteq [0,\infty)$.}\]
To see that $\seq{y}$ is bounded, we note that \eqref{equ:mu0} produces the finite upper bound:
\begin{align*}
 \obv(\mu_0) & \ge \limsup_n \E[V(y_n\tfrac{d\Q_n^r}{d\P})] + y_n \langle \Q_n,B \rangle \\
 &\ge \limsup_nV(y_n) + y_n\essinf B.
\end{align*}
The first inequality follows from \eqref{equ:rep} and the minimizing property of the sequence $\seq{\mu}$ whereas the last inequality follows from the non-increasing property of $V$ and Jensen's inequality. Because $V'(\infty)=0$ and $\essinf B>0$ the boundedness property of $\seq{y}$ follows.
}

Because the
finitely-additive probabilities $\seq{\QQ}$ belong to the weak$^*$-compact set
$\overline{\sM}^*$, we can conclude that $\seq{\mu}$ admits a
weak$^*$-convergent subnet $\mu_\alpha$ such that $\mu_{\alpha} \to \mu$, for
some $\mu\in\sD$. The functional $\obv$ is lower semicontinuous and we get
$$
 \valv(B) = \lim_\alpha \obv(\mu_\alpha ) \ge  \obv(\mu).
$$
Therefore, $\mu$ is a minimizer over $\sD$ and we have  $\opt\neq \emptyset$.

\medskip

Next, we show that all $\mu\in\opt$ have the same regular part. For
that, suppose that $\mu_1^r \ne \mu_2^r$. Then, $\mu = \tot \mu_1+\tot
\mu_2 \in \sD$ and by \eqref{equ:rep} we have
\[ \tot \obv(\mu_1) + \tot \obv(\mu_2) =
\tot \ee{V(\tRN{\mu_1^r}{\PP})} +
\tot \ee{V(\tRN{\mu_2^r}{\PP})} + \scl{\mu}{B} > \obv(\mu),\]
{by the strict convexity of  $V$. However,} this is in contradiction with the minimality of $\mu_1$ and $\mu_2$.

\medskip

To see that $\hY\neq 0$ we argue by contradiction and suppose that $
\PP[\hY=0]=1$. In that case $V(0) < \infty$ and, so, thanks to  Jensen's
inequality, we have $\obv(\mu)<\infty$ for all $\mu\in\sD$. In particular, we
have for some $\QQ \in \sM$ and $\hat{\mu} \in \opt$
$$
\obv(\mu^{\eps})<\infty, \quad \text{where }\quad
 \mu_\eps := \eps \QQ + (1-\eps) \hat{\mu},\quad \eps \in [0,1].
 $$
Because the regular-part functional is additive we have
$$
\mu_\eps^r = \eps \Q +(1-\eps)\hat{\mu}^r = \eps \Q.
$$
Therefore, $\hat{\mu}$'s minimality produces
\begin{align*}
\EE[ V(  \eps \tfrac{d\Q}{d\P} ) ] + \scl{\mu_{\eps}}{B}&=  \obv(\mu_{\eps})
\ge \obv(\hat{\mu})
= V(0) + \langle \hat{\mu},B\rangle.
 \end{align*}
Fatou's lemma then implies
$$
\langle \Q-\hat{\mu},B\rangle \ge \liminf_{\eps\downto 0}\oo{\epsilon}\Big
(V(0)-\EE\left[ V\left(  \eps \tfrac{d\Q}{d\P} \right) \right]\Big)  = -V'(0) = +\infty.
$$
This is a contradiction because $B\in\mathbb{L}^\infty(\PP)$ ensures that the left-hand-side is finite.

 \medskip

From the above, we know that $\hat{\mu}(\Omega)$ is uniformly bounded
over $ \opt$. Therefore, the weak$^*$-closed set $\opt$ is norm bounded and
the compactness property of $\opt$ follows from the Banach-Alaoglu theorem.

{Finally, to establish the strong duality property, we define the nested
sequence of weak$^*$-compact dual sets
$$
\sD_n := \{\mu \in \sD: || \mu || \le n\},\quad n\in \N,
$$
as well as the primal set
$$
\sC := (\sK-\lzer_+)\cap \linf = \sets{X\in\linf}{ \scl{\QQ}{X}\leq 0\text{  for all } \QQ\in\sM }.
$$
For a proof of the last identity see, e.g., Corollary 3.4(1) in \cite{LarZit12}.
As a consequence, we have the following identity for $X\in\mathbb{L}^\infty(\PP)$
\begin{align*}
\lim_{n\to\infty}\sup_{\mu\in\sD_n} \langle \mu,X\rangle = \begin{cases} 0, & X\in\sC, \\ +\infty, & X\not\in\sC. \end{cases}
\end{align*}
The minimax theorem (see, e.g., Theorem 2.10.2, p.~144 in \cite{Zal02}) can then be used to produce
\begin{align*}
\valv(B) &= \lim_{n\to\infty} \inf_{\mu\in \sD_n} \sup_{X\in\mathbb{L}^\infty(\PP)}\Big(\E[U(X+B)] - \scl{\mu}{X}\Big)\\
&=  \sup_{X\in\sC} \E[U(X+B)].
\end{align*}
The monotone convergence theorem ensures that this expression equals
the primal value function $\valu(B)$.}
\end{proof}

\section{Conditional Davis prices}
\label{sec:davis}
\begin{definition}
For $B\in\linfpp$, a random variable $ R\in\linf$ is said to be
$B$-\define{irrelevant}, denoted by $R\in\sI(B)$, if
  \begin{equation}%
\label{equ:irr}
    \begin{split}
    \valu(B+\eps R) \leq \valu(B), \quad
    \forall  \eps\in\R.
    \end{split}
\end{equation}
   \QEDB
\end{definition}
\begin{remark}
\label{rem:sI}
The function $\valu$  finite-valued at $B$, as well as in a
$\linf$-open ball around $B$.
Therefore, both sides of \eqref{equ:irr} are real-valued for
small enough $\eps$. Thanks to the concavity of $\valu$, the right-hand-side of \eqref{equ:irr} may only take the negative infinite value for large
values of $\eps$. Therefore, for $R\in\sI(B)$, it is enough to
check \eqref{equ:irr} only for $\eps$ in a neighborhood of $0$.
\end{remark}

\begin{lemma}
\label{lem:lin} $\sI(B)$ is a nonempty, weak$^*$ closed linear subspace in $\linf$.
\end{lemma}
\begin{proof}
The function $\valu$ is concave at $B$, so $\sI(B)$ is
the set of those directions $ R$ with the property that the directional
derivative of $\valu$ in directions $ R$ and $- R$ are nonpositive. In other
words, we have
\[
\sup_{\mu\in\partial \valu(B)} \scl{ R}{\mu} \leq  0 \eand
\sup_{\mu\in\partial \valu(B)} -\scl{ R}{\mu} \leq  0,\]
where $\partial \valu(B) \subseteq \ba(\P)$ is the super-differential of
$\valu$.  Therefore, $\sI(B)$ is the annihilator of
$\partial\valu(B)$, i.e.,
\[ \sI(B) = \sets{  R \in \linf}{ \scl{\mu}{ R}=0 \text{ for all }
\mu\in \partial \valu(B)},\]
which implies the statement.
\end{proof}

The following definition is due to Mark Davis and originates in \cite{Dav97}:

\begin{definition}
\label{def:Davis}
A number $p\in\R$ is said to be a $B$-\define{conditional Davis price} (or a
$B$-\define{marginal utility-based price})  and simply a \define{conditional
Davis price}  if $B$ is clear from the context, for a  payoff $ \vp\in\linf$ if \[
\vp-p \text{ is $B$-irrelevant.}\]
The set of all $B$-conditional Davis prices of $ \vp$ is denoted by $P( \vp|B)$.
   \QEDB
\end{definition}
\noindent Consequently, $p \in P( \vp|B)$ if and only if
\begin{equation}
\label{e.davis}
\valu(B+\eps (\vp -p)) \leq \valu(B), \quad
\forall  \eps\in\R.
\end{equation}

\section{marginal utility-based prices}
\label{s.new}
In this section we recall the definition
of marginal utility-based prices and
make the connection
to the conditional Davis prices defined
in the previous section.
We start with a definition
given in \cite{HugKra04}.
Let the derivative payoff $\xi$ be a
$\R^d$-valued, bounded, $\sF_T$ measurable
random variable.  Recall that
the function $v$ is defined in \eqref{intro1}.
\begin{definition}[Remark 1., p.~849 in \cite{HugKra04}]
\label{def:UBP}
A vector $p_0 \in \R^d$ is said to be a
{\em{marginal utility-based price} of $\xi$ at $(x_0,q_0) \in \R^{d+1}$} if
$$
v(q\cdot (\xi-p_0) ; x_0+q_0 \cdot \xi) \le
v(0 ; x_0+q_0 \cdot \xi), \quad
\forall  q \in \R^d.
$$
   \QEDB

\end{definition}

In our context,
 $B:= x_0+q_0 \cdot \xi$
and the investor
prices  units of $\xi$ in addition to $q_0$.
As observed in
 \cite{HugKra04}, to study these
 prices it is convenient
 to introduce a finite-dimensional
 value function.  Then,
 marginal utility prices
 can be expressed
 as  sub-differentials of this concave function.
 Indeed,  consider
the value function $u(q,x)$ defined on $\R^{d+1}$ by
\begin{equation}
\label{e.val}
u(q,x):=
\valu(x+ q\cdot \xi)= \sup_{X \in \sK}\, \Bee{ U\Big(x+ q\cdot \xi+ X\Big) }.
\end{equation}
Under our standing assumptions,
$u$ is a proper concave function
on $\R^{d+1}$.  Moreover,
if there is no gains process $X \in \sK$ such that
$x+ q\cdot \xi+ X \ge 0$, then the value
function $u$ is by definition
equal to minus infinity.

The elementary connection between marginal-utility based
prices and the sub-differential of $u$
in the sense of convex analysis is
given in Remark 1 in \cite{HugKra04}, Equation (3.11) in \cite{HugKraSch05},
 and Equation (24) in \cite{KS06}.
 We re-state it here for future reference.  First, we note that
at any $(q_0,x_0)$ in the interior of
 the domain of $u$, the set of
 sub-differentials is non-empty
 and compact.  Moreover, the second
 component of any $(z_q,z_x) \in \partial u(q_0,x_0)$
 satisfies $z_x>0$.

 \begin{lemma}[\cite{HugKra04},  \cite{HugKraSch05},  and \cite{KS06}]
 \label{lem:sub}
Let $y_0:=(q_0,x_0)\in \R^{d+1}$ be in the interior of
 the domain of $u$.  Then, $p_0 \in \R^d$ is a marginal
utility based price of $\xi$ at $y_0$ if and only if
$$
p_0 =\frac {z_q}{z_x},
$$
for some  $ (z_q,z_x)  \in \partial u(y_0)$.
 \end{lemma}
 \begin{proof}
 By definition $p_0 \in \R^d$ is a
utility based price of $\xi$ at $y_0$
if and only if
$$
u(q_0+ q,x_0-  q \cdot p_0) \le u(q_0,x_0),
\quad \forall q \in \R^d.
$$
We define
$$
f(q):= u(q_0+ q,x_0-  q \cdot p_0),\quad q \in \R^d.
$$
Then, $f$ is a concave function and
$p_0 \in \R^d$ is a marginal
utility based price of $\xi$ at $y_0$ if and only if
$0 \in \partial f(0)$.
Moreover, the sub-differential of $f$ is
connected to the sub-differential of $u$ by,
$$
\partial f(0)= \left\{
 - z_x p_0 + z_q \in \R^d\ :\
 \ z=(z_q,z_x)  \in \partial u(y_0) \right\}.
 $$
Therefore, $0 \in \partial f(0)$ if and only
if there exists $ (z_q,z_x)  \in \partial u(y_0)$
such that $ - z_x p_0 + z_q=0$.
 \end{proof}

\subsection{Conditional Davis prices and  marginal utility-based prices}
\label{ss.i}

In this subsection, we show that
conditional Davis prices defined
above can be seen as the projection of
marginal utility-based prices from the previous section
at an appropriately chosen point.
For given $B\in \linf_{++}$ and $\vp \in \linf$, we let $\sZ(\vp |B)$
be the set of all marginal utility-based prices $p_0 \in \R^2$
for the random variable $\xi:=(B,\vp)$
at the point $x_0:=0$, $q_0:=(1,0)$.
With these parameter choices
$p_0 \in \sZ(\vp|B)$ provided that
we have
\begin{equation}
\label{e.def}
\valu(B+q\cdot ((B,\varphi) -p_0))  \le
\valu(B),\quad \forall  q \in \R^2.
\end{equation}
Next, we show that the projection
of $\sZ(\vp|B)$ onto its second
component is the set of
conditional Davis prices $P(\vp|B)$ from Definition \ref{def:Davis} above.
Also, because $B \in \linfpp$ and $\vp \in \linf$, the point $(x_0,q_0)=
(0, (1,0))$ is in the
interior of the domain of
$u$ defined in \eqref{e.val}.

\begin{lemma}
\label{l.projection}
For $B\in \linf_{++}$ and $\vp \in \linf$ we have
$$
P(\vp |B) = \left\{ p\in \R \ :\ \exists \
p_0 \in \sZ(\vp |B) \ \
{\text{such that}}\ \
p= p_0 \cdot (0,1) \right\}.
$$
\end{lemma}
\begin{proof}
Let $(p_B,p) \in \sZ(\vp |B)$.
We use \eqref{e.def} with $q:=(0,\eps)$. The result is,
$$
\valu(B+\eps( \varphi-p))  \le
\valu(B),
$$
for all $\eps \in \R$. In view of \eqref{e.davis}, $p\in P(\vp|B)$.

To prove the converse, fix $p\in P(\vp|B)$.
Then, by \eqref{e.davis}
$$
u((1,\eps), -\eps p) =
\valu(B+\eps( \varphi-p))  \le
\valu(B) =
u((1,0), 0),
\quad
\forall \eps \in \R.
$$
Set
$$
g(\eps):= u((1,\eps), -\eps p).
$$
Then, $0 \in \partial g(0)$.  Also,
as in the proof
of Lemma \ref{lem:sub}, we have
$$
\partial g(0)= \left\{
- z_x p + z_q \cdot (0,1)  \in \R :\
 \ z=(z_q,z_x)  \in \partial u((1,0),0) \right\}.
 $$
Hence, there exists $z \in \partial u((1,0),0)$ such that
$$
0=- z_x p + z_q \cdot (0,1).
$$
We define $p_0:= z_q /z_x \in \R^2$ and use Lemma \ref{lem:sub} to arrive at
$p_0 \in \sZ(\vp|B)$.
It is also clear that
$p=p_0 \cdot (0,1)$.
\end{proof}

\begin{remark}\label{rmk:newMete}

In our context, we are given an {\em{endowment}}
 $B$ and a derivative with payoff $\vp$ (both $B$ and $\vp$ pay off at time $T$).  
 Our goal is to study marginal-utility based
 prices of $\vp$ conditioned on the
 fact that an endowment $B$ is given.  Clearly,
 one does not price the endowment $B$.
 Hence, the appropriate price
 is a projection of the set of marginal-utility based
 prices onto its second component
 with $\xi:=(B,\vp)$ at the
 points $x_0:=0$ and $q_0:=(1,0)$.
 In Lemma \ref{l.projection},
 we proved that these two approaches are equivalent.

We can also use the above
notation to summarize the related literature as follows:

\begin{enumerate}
\item Definition \ref{def:UBP} and Lemma \ref{lem:sub} are from \cite{HugKra04}.  \cite{HugKra04} study only the utility maximization problem and do not study marginal utility-based prices beyond their Remark 1.

\item Lemma \ref{lem:sub} can also be found in \cite{HugKraSch05}. Furthermore, when $B:=x>0$ is constant, \cite{HugKraSch05} provide a growth condition on the claim's payoff $\varphi$ which ensures uniqueness of marginal utility-based prices and exemplify that such prices can fail to be unique. In the case when $B:=x>0$ is constant, Theorem \ref{thm:derb} below supplements the results in \cite{HugKraSch05} with formulas for the two endpoints describing the non-trivial interval of marginal utility-based prices. We stress that when $B$ is unspanned, the results in \cite{HugKraSch05} do not apply. Example \ref{CSW} below illustrates that there can be major differences between the two cases: (i) $B:=x>0$ is constant and (ii) $B$ is unspanned.

\item \cite{KS06} use the growth condition from \cite{HugKraSch05} mentioned in (2) above which ensures uniqueness of the marginal-utility based prices. \cite{KS06} linearly expand the marginal-utility based price from the base case $B:=x>0$ constant. Our analysis differs in three crucial ways from \cite{KS06}: (i) we allow for non-uniqueness even when $B:=x>0$ constant,
(ii) we allow for $B$ being unspanned, and (iii) we do not perform an asymptotic expansion in small quantities $q$ of the claim's payoff $\varphi$ but we instead provide closed-form expressions  for the interval of  marginal utility-based prices in Theorem \ref{thm:derb} below. These non-trivial interval end-points are explicitly calculated in the two examples in Section \ref{sse:two}.
\end{enumerate}
\end{remark}

While the notion of pricing in Definition \ref{def:UBP} is consistent with the existing literature, no prior results cover the case where the investor's endowment $B$ is unspanned.

\section{Characterization of Conditional Davis Prices}
\label{sec:characterization}

\subsection{A dual characterization} The dual characterization of the set
of conditional Davis prices in Theorem \ref{thm:Davis} below rests on the following,
simple, lemma:
\begin{lemma}
\label{lem:d-char}
A random variable $ R\in\linf$ is $B$-irrelevant if and only if
     \begin{equation}
     \label{equ:V=V}
     \begin{split}
        \inf_{\mu \in \sD} \Big(
       \obv(\mu) + \abs{\scl{\mu}{ R}} \Big)=
        \inf_{\mu \in \sD}
       \obv(\mu).
     \end{split}
     \end{equation}
\end{lemma}
\begin{proof}
Because $\sI(B)$ is a vector space, we can scale $ R$ so that, without loss
generality, we can assume that $B \pm  R \in \linfpp$.
Then, by the
minimax theorem (see Theorem 2.10.2, p.~144 in \cite{Zal02}), we have
\begin{align*}
     \inf_{\mu \in \sD}
     \Big(  \obv(\mu) + \abs{\scl{\mu}{ R}} \Big)&=
    \inf_{\mu\in\sD}
    \sup_{\abs{\eps} \leq 1 } \Big( \obv(\mu) + \eps \scl{\mu}{ R}\Big)
    \\
    &= \sup_{\abs{\eps} \leq 1 }
    \inf_{\mu\in\sD}
    \Big( \obv(\mu) + \eps \scl{\mu}{ R}\Big) =
    \sup_{\abs{\eps}\leq 1} \valu( B + \eps  R ).
\end{align*}
The same equality with $ R=0$, implies that \eqref{equ:V=V} is equivalent
to \[ \valu(B) = \sup_{\abs{\eps}\leq 1} \valu(B+\eps  R)\]
which is,
in turn, equivalent to $ R \in \sI(B)$.
\end{proof}
By Lemma \ref{lem:minimizer}, we have $\mu(\Omega)>0$ for each $\mu\in
\opt$. Therefore, the family
\begin{align}
\label{equ:optz}
\optz := \sets{ \oo{\mu(\Omega)} \mu}{ \mu \in \opt}
\end{align}
is a well-defined nonempty family of finitely-additive probabilities. We now
have everything set up for our main characterization of conditional Davis prices:
\begin{theorem}
\label{thm:Davis} For $\vp \in \mathbb{L}^\infty(\PP)$ the following two statements are equivalent
\begin{enumerate}
\item $p\in P( \vp|B)$, i.e., $p$ is a $B$-conditional Davis price of $ \vp$.
\item $p = \scl{\QQ}{ \vp}$,  for some $\QQ\in \optz$.
\end{enumerate}
In particular, $P( \vp|B)$ is a nonempty compact subinterval of $\R$.
\end{theorem}
\begin{proof}\
{${(1)\Rightarrow (2)}$:} The first part of the
proof of Lemma \ref{lem:minimizer} applies to
the functional $\mu \mapsto \obv(\mu) +
\abs{\scl{\mu}{\vp-p}}$, and we can conclude that
it admits
a minimizer $\hmu$. By Lemma \ref{lem:d-char}, the same $\hmu$
must minimize the functional $\mu\mapsto \obv(\mu)$, as well, and, so,
$\hmu\in\opt$ and $\scl{\hmu}{ \vp-p}=0$.

\medskip

{\bf $(2)\Rightarrow (1)$:}
Suppose that $p$ is such that
$\scl{\mu^*}{\vp-p}=0$, for some $\mu^*\in \optz$. Then, for any $\mu$, we
have
\[ \obv(\mu^*) + \abs{\scl{\mu^*}{\vp-p}} = \obv(\mu^*)
\leq \obv(\mu) \leq \obv(\mu) + \abs{\scl{\mu}{\vp-p}},\]
and Lemma \ref{lem:d-char} can be used.

\medskip

Finally, Lemma \ref{lem:minimizer} ensures that $\opt$ is weak$^*$-compact  and the last claim follows.
\end{proof}

\subsection{First consequences}
A reinterpretation in the setting of portfolios with convex constraints
leads to the following dual characterization:
\begin{corollary}
\label{cor:approx}
Suppose that $U$ is reasonably elastic.
Then, for each constant $c\geq 0$ and each $ R\in\linf$ we have
\[
\inf_{\mu\in \sD} \Big( \obv(\mu) + c \abs{\scl{\mu}{ R}}\Big) =
\inf_{y\geq 0, \QQ\in \sM}\Big( \obv(y \QQ) + c \abs{\scl{y
\QQ}{ R}}\Big).\]
\end{corollary}

\begin{proof}
 Let $\sC:= (\sK-\lzer_+)\cap \linf$ and let $\sC'$ be the family of all
 random variables $X'\in\linf$ of the form
 \[ X' = B+X + \eps  R, \ewhere X\in\sC,\quad \eps \in [-c,c].\]
The support function $\alpha_{\sC'}$ for the set $\sC'$ is then given by
 \begin{align*}
 \alpha_{\sC'}(\mu) &= \sup_{X' \in \sC'}  \scl{\mu}{X} \\& = \scl{\mu}{B} +
\sup_{X\in\sC,\ \eps \in [-c,c]} \Big( \scl{ \mu}{X} + \eps
\scl{\mu}{ R}\Big)\\ &=
\scl{\mu}{B} + c\abs{\scl{\mu}{ R}} + \begin{cases} 0, & \mu\in\sD, \\ +
\infty, & \mu\not\in\sD. \end{cases}
 \end{align*}
It follows that
\begin{align}
\label{equ:abv}
\inf_{\mu\in\sD} \Big( \obv(\mu) + c\abs{\scl{\mu}{ R}} \Big) =
\inf_{\mu\in \ba(\P)} \Big( \obva{0}(\mu) + \alpha_{\sC'}(\mu) \Big).
\end{align}
Moreover, the set $\sC$ is weak$^*$-closed by Theorem 4.2 in \cite{DS94}; hence, so is $\sC'$. Hence, the assumptions of Proposition 3.14, p.~686
of \cite{LarZit12} are satisfied (via Corollary 3.4, p.~679 in
\cite{LarZit12}) and, so, the infimum on the right-hand side of
\eqref{equ:abv} can be replaced by an infimum over $\sigma$-additive
measures.
\end{proof}
Our next two consequences of Theorem \ref{thm:Davis} provide a partial generalization
and an alternative method of proof for Theorem 3.1, p.~206 in \cite{HugKraSch05}.
\begin{proposition}\label{pro:HKS}
Suppose that $U$ is reasonably elastic and that
the dual problem \eqref{equ:dual} admits a
non-$\sigma$-additive optimizer.  Then there exists $A\in\sF$ such that $ \vp=\ind{A}$ has
multiple $B$-conditional Davis prices.
\end{proposition}

\begin{proof}
Let $\seq{\mu}$ be a minimizing sequence for the problem $\inf_{\mu\in\sD}
\obv(\mu)$.
By Corollary \ref{cor:approx} we can assume that each $\mu$ is countably
additive. Moreover, the argument of Lemma \ref{lem:minimizer} guarantees that
the sequence $\seq{\mu(\Omega)}$ is bounded. Therefore, $\seq{\mu}$ belongs to
a weak$^*$-compact subset of $\ba(\P)$. By extracting a further subsequence,
we may assume that the sequence of total masses $\mu_n(\Omega)$ converges
towards a positive constant $y>0$ (Lemma \ref{lem:minimizer} ensures that
$y\neq0$).

\medskip

We suppose first that $\seq{\mu}$ is not weak$^*$-convergent. Then, two of its
convergent subnets will have different limits, and both of these will be
elements of $\opt$ with the same total mass $y>0$. Hence, the set $\optz$ of
\eqref{equ:optz} is not a singleton, and, by Corollary \ref{thm:Davis}, there
exists $ \vp=\ind{A}$, with $A\in\sF$, with two different conditional Davis prices.

\medskip

On the other hand, suppose that $\seq{\mu}$ converges to $\hat{\mu}$ in the
weak$^*$-sense. Then we have $\hat{\mu}\in\opt$. Furthermore, by the
Vitali-Hahn-Sachs theorem (see \cite{DunSch88}, Corollary 8 on p.159) the limit
$\hat{\mu}$ is countably additive. Hence, the set $\opt$ will have at least
two different elements - one countably additive and one not. Then a random
variable $ \vp=\ind{A}$ with two different conditional Davis prices can be constructed  as
above.
\end{proof}

The next consequence of Theorem \ref{thm:Davis} gives a sufficient condition (analogous to that of
Theorem 3.1 on p.~206 of \cite{HugKraSch05}) for the uniqueness of conditional Davis prices. Before we state it, we recall that, under the condition of reasonable elasticity, \cite{CSW01} show there exists a process $\hpi \in \sA$ such that $\hX := \hpist + B$ satisfies
  \begin{align}
  \label{equ:hat-X}
    \EE[ U(\hX) ] = \valu(B)\text{ and }U'(\hat{X}) =  \frac{d\hat{\mu}^r}{d\P},
    \end{align}
where $\hat{\mu}\in\hat{D}(B)$. The random variable $\hat{X}$ is $\PP$-a.s. unique with this property.

\begin{corollary}Suppose $U$ is reasonably elastic and that
$\abs{\vp} \leq c \hat{X}$, for some constant $c\ge0$, where $\hat{X}$ is as in \eqref{equ:hat-X}.
Then the set $P(\vp|B)$ of $B$-conditional Davis prices for $\vp\in\mathbb{L}^\infty$ is a singleton.
\end{corollary}
\begin{proof}
 In view of Lemma \ref{lem:minimizer} and Theorem \ref{thm:Davis}, it will
 be enough to show that $\scl{\hat{\mu}^s}{\hat{X}}=0$, for each $\hat{\mu}
 \in \hat{\sD}(B)$. This, in turn, follows directly from the first part of Equation (4.7) in \cite{CSW01}.
 \end{proof}

\section{Directional derivatives of the primal value function}
\label{sec:directional}

Our next task is to study directional differentiability of the primal utility-maximization
value function $\valu$ defined by \eqref{equ:u}. Its relevance in the context of Davis pricing has been noted
by several authors (including Davis in \cite{Dav97}), and we use the obtained
results in the later sections to give a workable characterization
of the interval of conditional Davis prices. First we show, by means of an example,
that smoothness - even in the most ``benign'' directions - cannot be expected
in general. Then we give a characterization of the directional derivative in terms of a
linear control problem.
We hope that both our counterexample and the later characterization hold
some independent interest outside of the context of Davis pricing.

\subsection{An example of nonsmoothness} Our next example shows that the
set $\hat{\sD}(B)$ of dual minimizers may contain measures with different total
masses. In other words, $\hat{\mu}(\Omega)$ may not be constant over $\hat{\mu} \in \hat{\mathcal{D}}(B)$. Consequently,
$\valu$ may fail to be differentiable even in ``constant directions'' in the sense that
$\eps \to \valu(B+\eps)$ may fail to be differentiable at $\eps :=0$.
Once we introduce the concept of
unique superreplicability in the next section, we will see how it can be
used to regain  differentiability in certain cases of interest.

For simplicity and concreteness, we base the example on Example 5.1'
in \cite{KS99}, and use the following notation and conventions: All random variables $X$ will be defined on the
sample space $\Omega:=\N_0$, and we write $X_n$ for $X(\set{n})$. Countably-additive
measures are identified with sequences in $\ell^1_+$ and for $\QQ=(q_n)\in \ell^1_+$ we write $\scl{\QQ}{X}$ for $\sum q_n X_n$ whenever $X=(X_n)\in \ell^\infty$.

\begin{example}\label{CSW} We start by recalling the elements
of (a special case of) the one period Example 5.1'
in \cite{KS99} where $\Omega := \mathbb{N}_0$ and $\PP=(p_n)$ with
$$
p_0 := \tfrac34, \quad p_n:= \tfrac{2^{-n}}{4} \efor n\in\N.
$$
The one-period stock-price increment $\Delta S=(\Delta S_n)$ is defined as follows
\begin{align*}
	\Delta S_0 := 1 \eand \Delta S_n := \tfrac{1-n}{n} \efor n\in\N.
\end{align*}
With $U:=\log$, the primal problem is defined by
\begin{align*}
	 u(x) := \sup_{\pi\in [-x,x]} \EE[ U(x+\pi \Delta S)],\quad x>0.
\end{align*}
Let $\sQ$ denote the set of all finite martingale measures, i.e.,
\begin{align*}
	\sQ := \sets{ \QQ\in\ell^1_{+}}{ \scl{\QQ}{\Delta S}=0 },
\end{align*}
and let $\sM := \sets{\QQ\in\sQ}{\scl{\QQ}{1}=1}$.
Because $V(y) = -1 - \log(y)$, the dual problem is given by
\begin{align*}
	 v(y) &:= \inf_{\QQ\in\sM} \EE[  V(y \tfrac{d\QQ}{d\PP})] =
	 V(y) + v^*,
	  \ewhere v^*:=\inf_{\QQ\in\sM} \EE[ - \log(\tfrac{d\QQ}{d\PP})].
\end{align*}

We will start by showing that no minimizing sequence $(\QQ^N)_N\subset \sM$ for $v^*$
 (equivalently, for $v(y)$) can
be weakly convergent in the sense that $\langle \QQ^N,f\rangle$ cannot converge for all test functions $f\in\ell^\infty$. It is a consequence of
the Vitali-Hahn-Sachs Theorem (see, e.g., Corollary 8, p.159 in \cite{DunSch88})
that $\ell^1$ is weakly sequentially complete, so any weakly convergent sequence
is necessarily weakly convergent in $\ell^1$.
Therefore, any weak limit of any minimizing sequence $(\QQ^N)_N$ must also
belong to $\sM$ and is, therefore, a minimizer for $v^*$. However,
this would contradict the strict supermartingale property of the dual
$\log$-optimizer shown in Example 5.1' in \cite{KS99}.

As a consequence of the above, for a given minimizing sequence $(\QQ^N)_N$, there
exists a random variable $H \in \ell^{\infty}$ such  that
\begin{align}
\label{equ:qN}
\scl{\QQ^N}{H} \text{ does not converge in $\R$ as $N\to\infty$.}
\end{align}
Because $\scl{\QQ^N}{1}=1$, for each $N$, we can assume that $H \geq 1$.
 Moreover, there exist two subsequences $(\QQ^{1,N})_N$ and $(\QQ^{2,N})_N$
 of $(\QQ^N)_N$ such that the limits
 \begin{align}
 \label{equ:seqs}
 	y_1=\lim_N \scl{\QQ^{1,N}}{H} \eand y_2=\lim_N \scl{\QQ^{2,N}}{H}
 	\text{ exist with  $y_1\ne y_2$. }
  \end{align}

With $H$ as above, we define $B:=1/H$ and a new stock price process with
increments
 $$ \Delta \tS := B\, \Delta S,$$ and then consider the log-utility maximization
 problem with the random endowment $B$ and the stock-price increments $\Delta \tS$.
 The associated dual problem\footnote{
It has been shown in \cite[Lemma 3.12]{LarZit12} that under the reasonable asymptotic elasticity condition, infimization over the set of countably-additive martingale measures - as opposed to its finitely-additive enlargement as in \cite{CSW01} - leads
to the same value function.}
  is given by
 \begin{align*}
 	\tv(y) &:=
 	\inf_{\tilde{\QQ}\in\tsM} \EE[  V(y \frac{d\tilde{\QQ}}{d\PP})] + y \scl{\tilde{\QQ}}{B}
 	\\
 	&= -1 +\inf_{\tilde{\QQ}\in\tsM} \Big( \EE[ - \log(y\frac{d\tilde{\QQ}}{d\PP})] +y \scl{\tilde{\QQ}}{B}
 	\Big)\\
 	&= -1 + \EE[ \log(B)] +
 	\inf_{\tilde{\QQ}\in\tsM} \Big( \EE[ - \log(y\frac{d\tilde{\QQ}}{d\PP} B)] +y \scl{\tilde{\QQ} B}
 	{1}
 	\Big)\\
	&=  \EE[ \log(B)] +
 	\inf_{\tilde{\QQ}\in\tsM} \Big( \EE[ V(y\frac{d\tilde{\QQ}}{d\PP} B)] +y \scl{\tilde{\QQ} B}
 	{1}
 	\Big),\quad y>0,
 \end{align*}
where 
\begin{align*}
	\tsQ := \sets{ \tilde{\QQ}\in\ell^1_{+}}{ \scl{\tilde{\QQ}}{\Delta \tS}=0 } \eand
\tsM := \sets{\tilde{\QQ}\in\tsQ}{\scl{\tilde{\QQ}}{1}=1}.
\end{align*}
Because $\tilde{\QQ} \in \tsQ$ if and only if $\QQ=\tilde{\QQ} B \in \sQ$, we have
\begin{align*}
	\inf_{y>0} \tv(y) 
	&= \EE[ \log(B)]  + \inf_{y>0} \inf_{\QQ \in \sM}
	\Big(\EE[V(y\frac{d\QQ}{d\PP})] + y \Big)\\
	&= \EE[ \log(B)]  + \inf_{y>0} \Big( v(y) + y \Big)
	\\&= \EE[ \log(B)] + \inf_{y>0} \Big( V(y) + y + v^*
	\Big) \\&= \EE[ \log(B)] + v^*.
\end{align*}

By using the minimizing sequences
$(\QQ^{1,N})_N$, and $(\QQ^{2,N})_N$ constructed above we  define the sequence of probability measures
\begin{align*}
	\tilde{\QQ}^{i,N} := \frac{ \QQ^{i,N} H}{ \scl{\QQ^{i,N}H}{1}}  \in \tsM \efor
	i=1,2.
\end{align*}
We can use \eqref{equ:seqs} and the fact that $(\QQ^{i,N})_N$, $i=1,2$, are minimizing sequences for $v^*$ to see
\begin{align*}
	\EE[ V(y_i \tfrac{d\tilde{\QQ}^{i,N}}{d\PP})]+y_i \scl{\tilde{\QQ}^{i,N}}{B}&=
	\EE[ V( \tfrac{y_i H}{\scl{\QQ^{i,N}H}{1}} \tfrac{d\QQ^{i,N}}{d\PP })]+\tfrac{y_i} {\scl{\QQ^{i,N} H}{1}}\\
	 &=
	 \EE[V(\tfrac{y_i H}{\scl{\QQ^{i,N}H}{1}} )] +
	 \tfrac{y_i}{\scl{q^{i,N}H}{1}}
	 - \EE[ \log( \tfrac{d\QQ^{i,N}}{d\PP})]\\&
	 \to \EE[ \log(B)] + v^*
	 \\&=\inf_{y>0} \tv(y).
\end{align*}
Clearly, $\tv(y_i)\ge \inf_{y>0} \tv(y)$, for $i=1,2$, which implies that
 $\tilde{\QQ}^{i,N}$ is a minimizing sequence for $\tv(y_i)$. Therefore,
 $\tv(y_1) = \tv(y_2) = \inf_{y>0} \tv(y)$ which implies that $\tv$ is constant
 on $[y_1,y_2]$. This, in turn implies, that the conjugate function to $\tilde{v}$
 fails to be differentiable at $0$ (indeed, the entire segment $[y_1,y_2]$ belongs to its superdifferential at zero).
\QEDB

\end{example}

\begin{remark}\
\begin{enumerate}
  \item The construction of the random endowment $B$ in Example \ref{CSW}
  above rests on the weak sequential completeness property of $\ell^1$ which, in fact, holds for any
  $\lone$-space. Example \ref{CSW} above is therefore generic in the sense that it can be applied to any
model which produces non-trivial singular components in the dual optimizer for
the log-investor (with constant endowment). This implies that there also exist
random endowments in the Brownian setting of Example 5.1 in \cite{KS99} which
produce a non-differentiable primal utility function.

\item Example  \ref{CSW} seems to contradict the claimed smoothness of the
primal value function stated in Theorem 3.1(i) in \cite{CSW01}\footnote{The authors first learned from Pietro Siorpaes about the potential lack of correctness of Remark 4.2 in \cite{CSW01}. We also refer the reader to Erratum \cite{CSW17} for further discussions.}:
 With the
notation from  Example \ref{CSW} we can define the primal utility function
\begin{align}\label{tildeu}
	 \tilde{u}(x) := \sup_{\pi\in \R} \EE[ U(x+\pi \Delta \tilde{S}+B)],\quad
	 x\in\R,
\end{align}
where we use the convention $\EE[  U(x+\pi \Delta \tilde{S}+B)]=-\infty$ if
 $\EE[  U(x+\pi \Delta \tilde{S}+B)^-]=+\infty$. Then $\tilde{u}$ is not
 differentiable at $x=0$ which is an interior point in $\tilde{u}$'s domain.
\end{enumerate}
\end{remark}

\subsection{A characterization via a linear stochastic control problem}
Even though the superdifferential of $\valu$ at $B$ consists of finitely-additive
measures related to the solution of the dual problem, it is possible to
give a characterization of directional derivatives without any recourse
to finite additivity. This is the most attractive feature of our linear
characterization in Proposition \ref{pro:var} below; however, as we shall see later,
it also leads to explicit computations in many cases. The price we pay is the
increased complexity of the linearized problem's domain.

Throughout the reminder of the paper we impose the following assumption,
where $\hat{X}$ is the primal optimizer characterized by \eqref{equ:hat-X},
and whose existence is guaranteed by the assumption of reasonable elasticity:
\begin{assumption}
\label{ass:U-prime} $U$ is reasonably
  elastic and there exists a constant $b>0$ such that
\begin{align}
\label{equ:ass-U}
\hX\,  U'\Bp{ (1-b) \hX}  \in \lone(\P).
\end{align}
  \QEDB
\end{assumption}
\begin{remark}
Assumption \ref{ass:U-prime} holds automatically if,
for example, $U$ belongs to the class of CRRA (power)
utilities
\[ U(x) = \tfrac{x^p}{p}, \efor p \in (-\infty,1)\setminus \set{0} \eor
U(x) = \log(x).\]
\end{remark}
Given the optimizer $\hpi\in\sA$ and the random variable $\hX$, we let $\Delta(\vp)
:= \cup_{\eps>0}\Delta^{\eps}(\vp)$ where
$\Delta^{\eps}(\vp)$
denotes the class of all
$\delta\in L(S)$, such that
   \begin{equation}
   \label{equ:Delta}
   \begin{split}
     \hpi + \eps \delta \in \sA \eand \hX + \eps (\vp+\dest) \geq  0.
   \end{split}
   \end{equation}
Because $\sA$ is a convex cone and $\hX\geq 0$, the family $\Delta^{\eps}(\vp)$ is nonincreasing in $\eps\geq 0$ in the sense
 \begin{align}
 \label{equ:mon-eps}
    \eps_1\leq \eps_2 \Rightarrow
   \Delta^{\eps_2}(\vp) \subseteq \Delta^{\eps_1}(\vp).
 \end{align}
Similarly, the family $\Delta^{\eps}(\vp)$ is nondecreasing in
$\vp\in\linf$ in the sense
 \begin{align}
 \label{equ:mon-vp}
    \vp_1 \leq \vp_2 \Rightarrow \Delta^{\eps}(\vp_1) \subseteq
   \Delta^{\eps}(\vp_2).
 \end{align}

\begin{proposition}
\label{pro:var}
Under Assumption \ref{ass:U-prime} we have for $\varphi\in\mathbb{L}^\infty(\P)$
\begin{align}
 \label{equ:var}
 \lim_{\eps \searrow 0}
 \oe \p{\valu(B+\eps  \vp)-\valu(B)} =
 \sup_{\delta\in\Delta(\vp)} \EE[ \hY \bp{\dest+ \vp}],
\end{align}
where $\hat{Y}:= \frac{d\hat{\mu}^r}{d\P}$.
\end{proposition}
\begin{proof} For small enough $\eps> 0$ we can find $\pi^{\eps}\in\sA$ such that $X^{\eps} = (\pi^{\eps}\cdot
S)_T + B + \eps \vp$ has the property that
\[ \EE[ U(X^{\eps}_T)] \geq \valu(B+\eps \vp) - \eps^2. \]
For such an $\eps>0$ we define
  \[ \delta^{\eps}= \oe \Big( \pi^{\eps} - \hpi \Big).\]
Since $\hpi+\eps \delta^{\eps} = \pi^{\eps}\in\sA$, the first part of
\eqref{equ:Delta} above holds.
To see that the second part of \eqref{equ:Delta} holds, we note that
$\hX+\eps\bp{\deest+\vp} = X^{\eps}$ and $\EE[ U(X^{\eps})] > -\infty$ which implies
$\hX+\eps\bp{\deest+\vp}\geq 0$. Therefore, we have $\delta^{\eps}\in
\Delta^{\eps}(\vp)$.
The concavity of $\valu$ then implies that
\begin{align*}
  \valu(B+\eps \vp)&\leq \EE[U(X^{\eps})] + \eps^2 \\&\leq  \EE[
  U(\hX)] +
  \eps \EE[ U'( \hX) (  \vp +  \deest)] + \eps^2 \\
  &\leq \valu(B) + \eps \sup_{\delta\in \Delta^{\eps}(\vp)} \EE[ \hY (\dest+ \vp)] +
  \eps^2
  \\ & \leq \valu(B) + \eps \sup_{\delta\in\Delta(\vp)} \EE[
  \hY(\dest+\vp)]+\eps^2.
\end{align*}
This produces the upper bound inequality
\[ \limsup_{\eps \searrow 0} \oe \Bp{ \valu(B+\eps \vp) - \valu(B)}
\leq \sup_{\delta\in\Delta(\vp)} \EE[ \hY(\dest+\vp)].\]

To prove the opposite inequality,
we pick $\eps_0>0$ and $\delta \in
\Delta^{\eps_0}(\vp)$, so that $\hpi+\eps_0 \delta\in\sA$ and
$\hX + \eps_0 D \geq 0$, where
\[ D= \dest+\vp.\]
Because $b>0$, we also have
\[ \hX + b \eps_0 D \geq
(1-b) \hX.\]
Therefore, for $\eps\in(0, \eps_1)$ with $\eps_1:=b \eps_0$ we have
 \begin{align}
 \label{equ:bnd}
    \hX + \eps D \geq (1-b)\hX>0 .
 \end{align}
The concavity of $U$ implies that for $\eps\in (0, \eps_1)$ we have
\[ U( \hX + \eps D) \geq
U(\hX) + \eps Y^{\eps} D \ewhere Y^{\eps}=U'\Bp{\hX+\eps D}.\]
Therefore, for $\eps\in(0, \eps_1)$ we  obtain
\[ \valu(B+\eps \vp) \geq \EE[ U(\hX+\eps  D)] \geq \valu(B) + \eps \EE[
Y^{\eps} D].\]

In order to pass $\eps$ to zero we note that \eqref{equ:bnd} gives us
 \begin{align}
 \label{equ:Fat}
    \Bp{ Y^{\eps}D }^- \leq U'\bp{(1-b)\hX} D^- \leq
  U'\bp{(1-b)\hX}\tfrac1{\eps_0} \hX,
 \end{align} which is integrable by assumption. The uniform bound in
 \eqref{equ:Fat} allows us to use Fatou's lemma together with $Y^\eps \to
 U'(\hat{X})=:\hat{Y}$, $\P$-a.s., to conclude that
\[ \liminf_{\eps\searrow 0}  \oe \Bp{ \valu(B+\eps \vp) - \valu(\vp) }
\geq \liminf_{\eps\searrow 0}\, \EE[ Y^{\eps} D] \geq \EE[ \hY
D].\qedhere \]
\end{proof}

The following example highlights the role strict local martingales play in
the linear optimization problem appearing in \eqref{equ:var}. The next
section identifies the key components which make this toy example work.

\begin{example} \label{ex:lmg}Let $(\Omega,\sF,\PP)$ be a probability space supporting two
independent Brownian motions $(Z,W)$ and we let $\prf{\sF_t}$ be their
augmented filtration up to some maturity $T>0$. We define the stock price dynamics to be
 \begin{align}
 \label{equ:dS_0}
   dS_t := S_t \big(\lambda_t dt + dZ_t\big),\quad S_0>0,
 \end{align}
 where the process $\ld$ is as in
\cite{DS98b} so that the
\emph{minimal martingale density}
$$
\sE(-\lambda \cdot Z)_t:= e^{-\int_0^t \lambda_u dZ_u -\frac12 \int_0^t \lambda^2_udu},\quad t\in[0,T],
$$
fails the martingale property even though the set $\sM$ of equivalent local
martingale measures is non-empty. As a consequence, Example 5.1 in \cite{KS99} shows that the log-investor's dual utility optimizer $\hat{Y}_t:=\hat{Y}_0\sE(-\lambda \cdot Z)_t$ is a strict local martingale.

We consider the simple case where $\hat{X}_0:=1$ and the payoff $\varphi$ is
constant. The fact that we are working with the $\log$-utility implies that
$\hat{Y}_0=1$, and  Remark 3.2 in
\cite{HugKraSch05} states that the unique Davis price of $\varphi$ is
$\varphi$ itself, a quantity
different from $\E[\hat{Y}_T\varphi]$.

For $\delta \in \Delta^\eps(\vp)$ we have
$$
(\delta \cdot S)_t \ge -\varphi -\tfrac 1\eps \hat{X}_t,\quad t\in[0,T].
$$
Thanks to the fact that $\hY\hat{X} = 1$, which is the standard myopic property
of optimizers in  logarithmic utility maximization,
the local martingale $\hY_t \bp{(\delta \cdot S)_t + \vp}$ is a lower
bounded by $-\tfrac1\eps$; hence, it is a supermartingale. Therefore, the
limit on the left-hand side of \eqref{equ:var} is bounded from above by
\begin{align}\label{simple_davis}
 \sup_{\delta\in\Delta(\vp)} \EE[ \hY_T \bp{\dest+ \vp}] \le \vp.
\end{align}
Because $\hat{Y}$ is a strict local martingale, we see that for any $\delta\in
\Delta(\vp)$ for which the local martingale $\hY_t (\delta \cdot S)_t$ is a
martingale the expression $\EE[ \hY_T \bp{\dest+\vp}]$ stays bounded away from the upper bound in \eqref{simple_davis}.
On the other hand, that upper bound
is attained at any
$\delta \in \Delta(\vp)$ which satisfies the requirement
 \begin{flalign*}
 && \hY_T (\delta \cdot S)_T = \vp (\hat{Y}\lambda\cdot Z)_T. && \Box
\end{flalign*}
\end{example}

\section{Uniquely superreplicable random variables}
\label{sec:unique}
While the linear control problem of Proposition \ref{pro:var}
provides a useful characterization of $\valu$'s directional derivatives, the linear problem seems
to be difficult to solve explicitly in full generality. The present section
outlines a relevant class of payoffs $\vp$ for which such a tractable solution
is, indeed, available. It involves the notion of unique superreplicability
similar to Condition (B1) in \cite{LSZ16}.
\begin{definition}
\label{def:minimal}
A random variable $\psi\in\mathbb{L}^\infty(\P)$ is said to be
\begin{enumerate}
\item \define{replicable}
if there exists a constant $\psi_0\in\R$ and $\pi_{\psi} \in \sA\cap (-\sA)$ such
  that
  \[ \psi = \psi_0 + (\pi_{\psi}\cdot S)_T.\]
\item
\define{uniquely superreplicable (by $\Psi$)} if $\Psi\in\mathbb{L}^\infty(\P)$ is replicable, $\Psi \geq \psi$, and
  \[
    x + \st{\pi} \geq \psi \Implies x+\st{\pi}\geq  \Psi \]
    for all
    $x\in\R$ and $\pi \in \sA$.
\end{enumerate}
\end{definition}
\begin{remark}\label{condB1}\
\begin{enumerate}
\item { The need to use uniformly bounded gains processes for replication purposes such as in Definition \ref{def:minimal}(1) has long been recognized; see, e.g., Definition 1.15 in \cite{SC02} and the first part of Remark 3.2 in \cite{HugKraSch05}.
}
\item The representation in Definition \ref{def:minimal}(1) of a replicable
claim $\psi$ in terms of $(\psi_0, \pi_\psi)$ is unique. Moreover, the process
$(\pi_{\psi}\cdot S)_t$ is a bounded $\QQ$-martingale for each $\QQ\in\sM$.
Consequently, because each $\mu \in \sD$ is the weak$^*$ limit of a net
$y_\alpha\Q_\alpha$ with $y_\alpha \in [0,\infty)$ and $\Q_\alpha \in \sM$, we
have
$$
\langle \mu, (\pi_{\psi}\cdot S)_t\rangle = \lim_\alpha y_\alpha \langle \Q_\alpha, (\pi_{\psi}\cdot S)_t\rangle = 0.
$$

\item Provided it exists, the random variable $\Psi$ in Definition \ref{def:minimal}(2) is unique. If $\Psi = \Psi_0 + (\pi_{\Psi}\cdot S)_T$ uniquely superreplicates $\psi$, we have
the representation
$$
\Psi_0 = \sup_{\Q\in\sM} \E^\Q[\psi].
$$

\item Unique superreplicability is scale invariant: If $\psi$ is uniquely
superreplicable by $\Psi$, then $\alpha \psi$ is uniquely superreplicable by
$\alpha \Psi$ for $\alpha \geq  0$.  It is also invariant under
translation by replicable random variables. In particular, replicable
random variables are uniquely superreplicable.
\end{enumerate}
\end{remark}

\begin{example}\label{example3}
Let $(\Omega,\sF,\PP)$ be a probability space supporting two
independent Brownian motions $(\beta,W)$ and we let $\prf{\sF_t}$ be their
augmented filtration up to some maturity $T>0$. With the set of all
pathwise $p$-integrable predictable processes denoted by $\sL^p$, we let
 $S$ be the It\^ o process
 \begin{align}
 \label{equ:dS-ex}
   dS_t := S_t\sigma_t \big(\lambda_t dt + d\beta_t\big),\quad S_0>0,
 \end{align}
where $\sigma, \lambda\in \sL^2$ are such that NFLVR
holds.

\medskip

We focus on payoffs of the form $\vp=\varphi(W_T)$, where
$\varphi:\R\to\R$ is a bounded Lipschitz function. To show that such $\varphi(W_T)$ is
uniquely superreplicable by the constant $\sup_a
\varphi(a)$, we start by assuming that
\[ x+ (\pi \cdot S)_T\ge \varphi(W_T)\text{ a.s., }\]
for some $x\in \R$ and some $\pi\in \sA$. Then, for each $t \in [0,T)$ we
  have
 \begin{align}
 \label{equ:esssup}
x+(\pi \cdot S)_t &\ge \esssup_{\Q\in\sM} \E^\Q[x+(\pi \cdot S)_T|\sF_t] \ge
\esssup_{\Q\in\sM} \E^\Q[\varphi(W_T)|\sF_t].
\end{align}
 Lemma \ref{lem:lone} below
 gives conditions under which the limit as $t\uparrow T$ of
 the right-hand side of \eqref{equ:esssup} equals $\sup_a \varphi(a)$.
 When these conditions are met, the continuity of the paths of
 the stochastic integral with respect to $S$ implies that
 $x+ (\pi\cdot S)_T \geq \sup_a \vp(a)$. This, in turn,   confirms
 that $\vp(W_T)$ is uniquely superreplicable by the constant $\sup_a
 \vp(a)$.
  \QEDB
\end{example}

\begin{lemma} \label{lem:lone} In the setting of Example \ref{example3}
above with $\varphi:\R\to\R$ bounded and Lipschitz, assume that there exists a nonnegative (deterministic) function
$f\in \lone([0,T])$  and a predictable process
$\nu^{(0)} \in\sL^2$ such that
\begin{enumerate}
  \item  $|\nu_u^{(0)}| \le f(u)$, for Lebesgue-almost all $u\in[0,T]$,
  $\PP$-a.s., and
  \item the stochastic exponential $Z^{(0)}_T := \EN( -\ld\cdot \beta -
  \nu^{(0)}\cdot W)_T$ is the Radon-Nikodym density of some
  $\QQ^{(0)}\in\sM$ with respect to $\PP$.
\end{enumerate}
Then
\begin{align}\label{conclusion}
\lim_{t\uparrow T}\esssup_{\Q\in\sM} \E^\Q[\varphi(W_T)|\sF_t]= \sup_a \vp(a).
\end{align}
\end{lemma}

\proof
For a bounded and predictable process $\delta$ we define the process
$Z^{(\delta)}$ by
$$
dZ^{(\delta)}_t := - Z^{(\delta)}_t\big(\lambda_t\, d\beta_t +
(\nu^{(0)}_t+\delta_t)\,
dW_t\big)\quad Z^{(\delta)}_0:=1.
$$
A simple calculation yields the following expression
\[ Z^{(\delta)}_T =  Z^{(0)}_T \EN(-\delta\cdot W^{(0)})_T,\]
where $W^{(0)}_t := W_t + \int_0^t\, \nu^{(0)}_u\, du$ is a
$\QQ^{(0)}$-Brownian motion. With $\EE^{(0)}$
denoting the expectation with respect to $\QQ^{(0)}$, we have
\[ \EE[ Z^{(\delta)}_T] = \EE^{(0)}[ \EN(-\delta\cdot W^{(0)})]=1,\]
where the last equality follows from the boundedness of $\delta$. Hence,
$Z^{(\delta)}$ is a (true) martingale and can be used as a density of a
probability measure $\QQ^{(\delta)}\in\sM$.

To proceed, we fix $t_0\in(0,T)$ and $a\in\R$ and  define
\begin{align*}
  \delta^{(a)}_t & := \oo{T-t_0} (W_{t_0}\inds{\abs{W_{t_0}}\leq 1/(T-t_0)} -a)
  \inds{t\geq t_0},\quad t\in [t_0,T], \\
  W^{(a)}_t &:= W_t + \int_0^t (\nu^{(0)}_u + \delta^{(a)}_u)\, du,\quad t\in [0,T].
\end{align*}
Then we have
   \begin{align*}
   W_T - a =
     W^{(a)}_{T}  -  W^{(a)}_{t_0} -  \int_{t_0}^T
     \nu^{(0)}_u\, du +  W_{t_0}\inds{\abs{W_{t_0}}> 1/(T-t_0)}.
   \end{align*}
  The process $W^{(a)}$ is a $\QQ^{(a)}$-Brownian motion, where $\QQ^{(a)}$
  is a short for $\QQ^{(\delta^{(a)})}$.
  Therefore, the bound
  $\abs{\nu^{(0)}}\leq f$ implies that
\begin{align*}
  \EE^{\QQ^{(a)}}\big[ \abs{ W_T - a} \big|\sF_{t_0}\big] \leq C(t_0)
\end{align*}
where
\[ C(t_0):=\sqrt{\tfrac{2(T-t_0)}{\pi}} + \int_{t_0}^T f_u\, du +
\abs{W_{t_0}}
  \inds{\abs{W_{t_0}}>1/(T-t_0)}. \]
With $L_{\vp}$ denoting the Lipschitz constant of $\vp$, we have
\[ \abs{ \EE^{\QQ^{(a)}}[ \vp(W_T) |\sF_{t_0}]  - \vp(a) } \leq
  L_{\vp}\, \EE^{\QQ^{(a)}}[ \abs{ W_T - a} |\sF_{t_0}] \leq L_{\vp} C(t_0).\]
Therefore,
    \begin{align*}
        \limsup_{t_0 \nearrow T} \esssup_{\QQ\in\sM} \EE^{\QQ}[ \vp(W_T) |
        \sF_{t_0}]
      &\geq \limsup_{t_0 \nearrow T}
      \EE^{\QQ^{(a)}}[ \vp(W_T) |\sF_{t_0}]\\ &  \geq \limsup_{t_0 \nearrow T}
      \Big(\vp(a) - L_{\vp} C(t_0)\Big) = \vp(a).
    \end{align*}
It remains to note that the left-hand side above does not depend on $a$
and that $\sup_{a} \vp(a)$ is a
trivial upper bound in \eqref{conclusion}.
\endproof

\begin{example}
\label{example3a} (Continuation of Example \ref{example3}) We continue
Example \ref{example3} by examining two cases in which Lemma
\ref{lem:lone} applies. In the first one we simply take $f:= 0$. That
can be done if and only if
the minimal martingale density $\EN(-\ld\cdot B)$ defines a
martingale which is the case in many popular models including the incomplete models developed in \cite{KimOmb96} and \cite{Kra05}.

In the second case, $\EN(-\ld\cdot B)$ is a strict local martingale
but NFLVR nevertheless still holds. A famous example of a model where this occurs is given in \cite{DS98b}. We present here a
time-changed version (using the standard logarithmic time transform $t\mapsto
-\log(1-t)$), as the original version in  \cite{DS98b} is defined on an infinite horizon.
In the notation of Example \ref{example3}, and with $T=1$,
we define the local martingales
$(\beta'_t)_{t\in [0,1)}$ and $(W'_t)_{t\in [0,1)}$ by
\[
\beta'_t := \int_0^t \oo{\sqrt{1-u}}\, d\beta_u \eand
W'_t := \int_0^t \oo{\sqrt{1-u}}\, dW_u, \quad t\in [0,1),\]
as well as the stopping times
\begin{align*}
\tau := \inf \{ t>0 : \sE(\beta') =1/2\} \eand
\sigma := \inf \{ t>0 :\sE(W') =2\}.
\end{align*}
With the processes
$(\ld_t)_{t\in [0,1]}$ and
$(\nu^{(0)}_t)_{t\in [0,1]}$ given by
\[ \ld_t :=
-\frac{1_{\sigma\land\tau}(t)}{\sqrt{1-t}},\quad \nu^{(0)}_t :=
-\frac{1_{\sigma\land\tau}(t)}{\sqrt{1-t}},
\]
it remains to apply the results of \cite{DS98b} to conclude that the NFLVR
condition is satisfied, but that the minimal martingale density
$\EN(-\ld\cdot \beta)$ is a strict local martingale. Our Lemma
\ref{lem:lone} applies because $\abs{\nu^{(0)}_t} \leq \oo{\sqrt{1-t}} \in
\lone([0,1])$.

We conclude the  by mentioning that Examples \ref{example3} and
\ref{example3a}, as well as Lemma \ref{lem:lone} will be used again in the
examples in Section \ref{sec:interval}.
 \QEDB

\end{example}

{The next example shows that it is quite easy to construct bounded payoffs $\psi$ which fail to be uniquely superreplicable.
\begin{example} We consider the following one period  model with three states:
$$
\Delta S := (1,0,-1)',\quad \psi := (-1,0,-1)'.
$$
The set of pairs $(x,\pi)$ for which $x+ \pi \Delta S\ge \psi$ is given by $x\ge 0$ and $\pi \in [-1-x,1+x]$. However, the corresponding set of gain outcomes $(x+\pi,x,x-\pi)'$ with $x\ge 0$ and $\pi \in [-1-x,1+x]$ does not contain a smallest element. Indeed, if $(a,b,c)$ is a smallest element, we would have $a\le -1, b\le0$, and $c\le -1$ but such an element $(a,b,c)$ is not the outcome of any gains process $x+ \pi \Delta S$ with $x\ge 0$ and $\pi \in [-1-x,1+x]$.
 \QEDB

\end{example}
}

The main technical result of this section is the following proposition:
\begin{proposition}
\label{pro:limit}
Under Assumption \ref{ass:U-prime}, suppose that
$-B$ and $-(B+\eps \vp)$ are uniquely superreplicable by $-\uB$
and $-(\ubev)$,   respectively,
for all $\eps>0$ in some neighborhood of $0$. Then, for each $\hmu\in\opt$,
  we have
 \begin{align}
 \label{equ:der-exp}
    \lim_{\eps \searrow 0} \oe \Bp{ \valu(B+\eps \vp) - \valu(B)} = \EE[\hY
   \vp] + \lim_{\eps \searrow 0} \oe \Bscl{ \hmu^s}{ \ubev - \uB},
 \end{align}
 where $\hY = \frac{d\hat{\mu}^r}{d\P}$.
\end{proposition}
\begin{proof}
For $\eps>0$ we let $x_{\eps}$, $x_0\in\R$ and $\pi_{\eps}, \pi_0 \in \sA\cap
(-\sA)$ be such that
\[ \ubev= \eps x_{\eps} + \eps \st{\pie} \eand \uB = x_0 + \st{\pi_0}.\]
Because $B$ is bounded away from zero and $\varphi\in \mathbb{L}^\infty(\P)$
we can consider $\eps>0$ so small that $x_{\eps},x_0>0$. For $\delta \in
\Delta^{\eps}(\vp)$ we have $\eps\delta + \hpi \in \sA$ and
\[\eps \dest + \hpist  \geq -\eps\vp - B.\]
Therefore, by the unique superreplicability of $\ubev$, we have
 \begin{align}
 \label{equ:poz}
0\le   x_{\eps} + \dest + \oe \hpist + \st{\pi_{\eps}} .
 \end{align}

Since $\uB$ uniquely superreplicates $B$ we have
$\uB + \hpist \geq 0$. Therefore, for any $\hmu\in\opt$, the first part of Equation (4.7) in \cite{CSW01} produces
 \begin{align}
 \label{equ:KL1}
 0 \leq  \scl{\hmu^s}{ \uB + \hpist} \leq
\scl{\hmu^s}{B+\hpist} = 0.
\end{align}
The second part of Equation (4.7) in \cite{CSW01} ensures $\scl{\hmu}{\hpist}=0$ and combining this with $ \scl{\hmu^s}{ \uB + \hpist} =0$ we see
\begin{align*}
\langle \hat{\mu}^r, \uB + (\hat{\pi}\cdot S)_T\rangle&=\langle \hat{\mu}, \uB + (\hat{\pi}\cdot S)_T\rangle\\&= \langle \hat{\mu},\uB\rangle.
\end{align*}
Because $\hat{Y} = \tfrac{d\hat{\mu}^r}{d\P}$ we obtain the representation
\begin{align}\label{KL2}
  \EE[ \hY \hpist]=\scl{\hmu^s}{\uB}.
  \end{align}

The property $\eps\delta + \hpi \in \sA$ produces $\scl{\hmu}{\eps\dest + \hpist}\leq 0$ and $\langle \hat{\mu},(\pi_\eps\cdot S)_T\rangle =0$, for each $\eps>0$. Therefore,
by  \eqref{equ:poz}, we find
\begin{align}
\label{equ:upper-m}
\begin{split}
&\EE[ \hY(x_{\eps}+ \st{\pi_{\eps}} + \st{\delta} + \oe \hpist )]
\\
& \leq
\scl{\hmu}{x_{\eps}+ \st{\pi_{\eps}} + \st{\delta} + \oe \hpist }  \\&\leq
\scl{\hmu}{x_{\eps}}.
\end{split}
\end{align}
To show that the upper bound in \eqref{equ:upper-m} above is attained
we pick
\begin{align}\label{deltaeps} \dele = (\tfrac{x_{\eps} }{x_0} - \oe) \hpi
+ \tfrac{x_{\eps}}{x_0} \pi_0 -
\pi_{\eps}.
\end{align}
Because $x_\eps>0$ and $x_0>0$ one can check that $\dele \in
\Delta^{\eps}(\vp)$. Then we have
 \begin{align*}
   \EE[ \hY(&x_{\eps}+ \st{\pi_{\eps}} + \sst{\dele} + \oe \hpist )] \\&=
   \EE[ \hY(x_{\eps}+ \tfrac{x_{\eps}}{x_0}\st{( \hpi + \pi_0)})]\\
      & = \tfrac{x_{\eps}}{x_0} \EE[ \hY( \uB + \st{\hpi})] \\ & =
   \tfrac{x_{\eps}}{x_0} \scl{\hmu}{\uB + \st{\hpi}}\\ & =
   \scl{\hmu}{x_{\eps}},
 \end{align*}
 where the last equality follows from  $\langle \hat{\mu},\uB\rangle=\langle
 \hat{\mu},x_0\rangle$ and $\langle \hat{\mu},(\hat{\pi}\cdot S)_T\rangle=0$.
Therefore, $\dele$ indeed attains the upper bound of \eqref{equ:upper-m},
and, so,
 \begin{align*}
    \sup_{\delta \in \Delta^{\eps}(\vp)} &\EE[ \hY( \dest + \vp)]
   = \\ &=
   \scl{\hmu}{x_{\eps}} - \oe\EE[\hY(\ubev)] - \oe \EE[ \hY
   \hpist] + \EE[ \hY \vp] \\
   &= \scl{\hmu^r}{\vp} +
   \oe\scl{\hmu^s}{\ubev -  \uB}.
 \end{align*}

The sets $\Delta^{\eps}(\vp)$ monotonically increase to $\Delta(\vp)$ as $\epsilon  \searrow 0$. This implies
\begin{align}
\label{KL1}
    \sup_{\delta \in \Delta^{\eps}(\vp)} \EE[ \hY( \dest + \vp)] \nearrow  \sup_{\delta \in \Delta(\vp)} \EE[ \hY( \dest + \vp)]
\end{align}
as $\eps  \searrow 0$. Because the left-hand-side of \eqref{KL1} equals  $\scl{\hmu^r}{\vp} + \oe\scl{\hmu^s}{\ubev -  \uB}$, we see that \eqref{equ:der-exp} holds by Proposition
\ref{pro:var}.
\end{proof}
A first consequence of Proposition \ref{pro:limit} is that the situation
encountered in Example \ref{CSW} cannot happen if $-B$ if uniquely superreplicable.
Indeed, the primal value function $\valu$ is then smooth in all replicable directions:
\begin{corollary}
\label{cor:42F2}
Suppose that Assumption \ref{ass:U-prime} holds, that $-B$ is uniquely
superreplicable and that $\vp$ is replicable.
Then the following two-sided limit exists
 \begin{align}\label{limit_replicable}
    \lim_{\eps \to 0} \oe \Bp{ \valu(B+\eps \vp) - \valu(B)} =
    \scl{\hmu}{\vp} \text{ for each } \hmu \in \opt.
 \end{align}
In particular, there exists a constant $y_B>0$ such that
\[ y_B = \hmu(\Omega)\text{ for each } \hmu \in \opt.\]
\end{corollary}
\begin{proof}
First observe that for replicable $\vp$ we have
$\ubev = \uB + \eps \vp$.
Then we can apply Proposition \ref{pro:limit} to both $\vp$ and $-\vp$ to see \eqref{limit_replicable}. The last claim follows by setting $\varphi=1$.
\end{proof}
Now that we know that all dual minimizers $\hmu\in\opt$ have the same
total mass, the following result follows directly from Corollary
\ref{thm:Davis} above.
\begin{corollary}
Suppose that Assumption \ref{ass:U-prime} holds and that $-B$ is uniquely
superreplicable. Then
each replicable $\vp\in\mathbb{L}^\infty(\P)$ has the unique $B$-condi-tionally Davis price $\langle \hat{\mu},\vp\rangle/\hat{\mu}(\Omega)$.
\end{corollary}

When $B$ is a constant (and more generally, when $B$ is replicable), it is
known that the product of the primal and dual optimizers is a martingale (see,
e.g., the discussion on p.911-2 in \cite{KS99}). When the dual optimizer is
only a finitely-additive measure, the following corollary may serve as a
surrogate. The result relies on \cite{KarZit03} where a positive
supermartingale deflator $\{\hat{Y}_t\}_{t\in[0,T]}$ is constructed from $\hat{\mu}\in
\hat{\sD}(B)$ (see Equation 2.5 in  \cite{KarZit03}).

\begin{corollary} Suppose that Assumption \ref{ass:U-prime} holds, that $-B$ is uniquely superreplicable by $-\uB$, and write
$$
\uB = x_0 + (\pi_0\cdot S)_T.
$$
Then the process
$$
\hat{Y}_t \Big(x_0 + \big((\pi_0+\hat{\pi})\cdot S\big)_t\Big),\quad t\in[0,T],
$$
is a nonnegative martingale where $\{\hat{Y}_t\}_{t\in[0,T]}$ is the supermartingale deflator corresponding to $\hat{\mu}\in \hat{\sD}(B)$.
\end{corollary}

\begin{proof}
From Theorem 2.10 in \cite{KarZit03} we know that the process in question
is a nonnegative supermartingale. Furthermore, also from \cite{KarZit03},
we have $\hat{Y}_T = \frac{d\mu^r}{d\P}$ and $\hat{Y}_0 \le \hat{\mu}(\Omega)$.
To obtain the  constant expectation property we use \eqref{equ:KL1} to get
\begin{align*}
\langle\hat{\mu},x_0\rangle &= \langle \hat{\mu}, (\hat{\pi}\cdot S)_T+ \uB\rangle
= \langle \hat{\mu}^r, (\hat{\pi}\cdot S)_T+ \uB\rangle\\
&= \E\left[\hat{Y}_T\big( (\hat{\pi}\cdot S)_T+ \uB\big)\right] \le \hat{Y}_0 x_0
\le  \hat{\mu}(\Omega) x_0,
\end{align*}
and the claimed martingale property follows.
\end{proof}

\section{The interval of Conditional Davis prices}
\label{sec:interval}

This section  closes the loop and gives an explicit expression
for the interval of conditional Davis prices under the assumption of unique superreplicability. We start
with a standard
characterization of conditional Davis prices in terms of perturbed value functions.
Given $B \in \linfpp$ and $\vp\in\linf$ (we do not impose any unique
superreplicability assumption on either, yet). We let the function $u:\R^2 \to [-\infty,\infty)$ be defined by
\[ u(\eps,x) := \valu(B+x+\eps  \vp),\]
and let its supergradient at $(0,0)$ be denoted by $\partial u(0,0)$.
\begin{lemma}
\label{lem:subgrad} Let  $\vp\in\linf$. If $(0,0)\in \partial u(0,0)$, then $P( \vp|B)=\R$. Otherwise,
\[ P( \vp|B) = \sets{ \delta/y}{ y \ne 0, (\delta,y) \in \partial u (0,0) }.\]
\end{lemma}
\begin{proof}
Thanks to the assumption that $B\in\linfpp$, $u$ is concave and finite-valued in some
neighborhood of $(0,0)$. Moreover Definition \ref{def:Davis} translates
into the following statement:
\[ p \in P( \vp|B) \text{ if and only if } u(0,0) \geq u(\eps, - \eps p)
\eforall \eps. \]  By concavity, this is
equivalent to the nonpositivity of the directional derivative of $u$, at
$(0,0)$ in the directions $(-1,p)$ and $(1,-p)$, i.e.
\[ \inf_{(\delta,y) \in \partial u (0,0)} -\delta + p y \leq 0\  \eand
\inf_{(\delta,y) \in \partial u (0,0)}  \delta  - py \leq 0.\]
By the convexity of the supergradient, this is equivalent to existence of a
pair $(\delta,y) \in \partial u(0,0)$ such that $p y = \delta $.
\end{proof}
\begin{theorem}
\label{thm:derb}
Suppose that Assumption \ref{ass:U-prime} holds and that
$-B$ and $-(B+\eps \vp)$ are uniquely superreplicable by $-\uB$
and $-(\ubev)$,   respectively, for all $\eps$ in some neighborhood of $0$.
The interval of $B$-Davis prices of $\vp$ is given by
 \begin{align}
 \label{equ:der-2}
   \oo{y_B} \EE[ \hY \vp] + \oo{y_B}\Big[ \lim_{\eps \searrow 0} \oe \Bscl{ \hmu^s}{ \ubev - \uB},
   \lim_{\eps \nearrow 0} \oe \Bscl{ \hmu^s}{ \ubev - \uB}
   \Big],
 \end{align}
 where $y_B$ is the common value of $\hmu(\Omega)$ for all $\hmu\in\opt$.
\end{theorem}
\begin{proof}
By Corollary \ref{cor:42F2}, the function $u$ is
differentiable in $x$ at $x=0$, with derivative $y_B$. The interval of
$B$-conditional Davis prices, according to Lemma \ref{lem:subgrad}, is given by
\[ \tfrac{1}{y_B} [ \partial_{\eps+} u(0,0), \partial_{\eps-} u(0,0)].\]
This, in turn, coincides with the expression in \eqref{equ:der-2} thanks to
Proposition
\ref{pro:limit}.
\end{proof}

\subsection{Two illustrative examples }\label{sse:two} We conclude by giving two
illustrative examples, both in an incomplete Brownian setting, of
situations where our results can be applied directly and lead to explicit
formulas for the non-trivial interval of conditional Davis prices.

Let $(\Omega,\sF,\PP)$ be a probability space supporting two
independent Brownian motions $(Z,W)$ and we let $\prf{\sF_t}$ be their
augmented filtration up to some maturity $T>0$. The set of all
pathwise $p$-integrable predictable processes is denoted by $\sL^p$ and
the
space finitely-additive measures which are $\PP$-absolutely continuous is
denoted ba$(\PP)$ so that ba$(\PP)$ can be identified
with $\mathbb{L}^\infty(\PP)'$.

In both examples, the stock-price dynamics are given by a one-dimensional
It\^ o process

 \begin{align}
 \label{equ:dS}
   dS_t := S_t\sigma_t \big(\lambda_t dt + dZ_t\big),\quad S_0>0,
 \end{align}
with processes $\sigma, \lambda\in \sL^2$. With more driving Brownian
motions than assets, this leads to a incomplete financial market.
Both examples will feature (an unspanned) contingent claim
paying out $\varphi(W_T)$ at time $T$, where $\varphi:\R\to\R$ is a
non-constant, bounded, and
continuous function.

The major difference between the examples is that in the first example the
illiquid random endowment degenerates ($B:=x$ for a constant $x>0$), while in the second example the random endowment $B$ is non-replicable. The first example illustrates that even when $B:=x>0$ is constant, our setting differs from that of \cite{KS06} because the corresponding Davis prices are non-unique whereas the growth condition placed on the claim's payoff in
\cite{KS06} always produces unique Davis prices (the growth condition used in \cite{KS06} originates from \cite{HugKraSch05}). In other words, the payoffs considered in the first example are not included in \cite{KS06}. The second example backs up the claim we made in both the abstract and in the introduction: When the endowment $B$ is non-replicable, the generic case is that Davis prices are non-unique.

\subsubsection{Example 1}
We adopt the setting used in Example \ref{example3} above which is based on \cite{DS98b}.  The endowment is taken to be $B:=x>0$ constant. It follows from Example \ref{example3}  that the interval of arbitrage-free prices for
$\varphi(W_T)$ is given by $(\underline{\varphi},\overline{\varphi})$ where
\begin{align}\label{arb_free_price_interval}
\underline{\varphi} := \inf_{a\in\R} \varphi(a), \quad\overline{\varphi} := \sup_{a\in\R} \varphi(a).
\end{align}
Our Theorem \ref{thm:derb} with $B:=x>0$ constant shows that the interval of
$\log$-investor's Davis' prices
for $\varphi(W_T)$ is given by
$[\underline{p},\overline{p}]$ where
$$
\underline{p} := \tfrac1{\hat{Y}_0}\E[\hat{Y}_T(\varphi-\underline{\varphi})]+
\underline{\varphi},\quad \overline{p}:=\overline{\varphi} -
\tfrac1{\hat{Y}_0}\E[\hat{Y}_T(\overline{\varphi}-\varphi)].
$$
Therefore,
since the function $\vp$ is not constant, we have
\begin{align*}
	\overline{p} - \underline{p} = (\overline{\varphi} - \underline{\vp})
	(1- \EE[ \hat{Y}_T]/ \hat{Y_0})>0.
\end{align*}

\subsubsection{Example 2}
In this example, we consider the Samuelson-model setting used in Section 2
in \cite{LSZ16} where the
stock price dynamics are given by \eqref{equ:dS} with both $\sigma_t:=\sigma>0$
and $\ld_t:=\ld>0$ being constants.
  Let $U(\xi):=\frac{\xi^\gm}\gm$, $\xi>0$, $\gm <1$, be a utility function in the
  ``power'' family,
  with constant relative risk-aversion parameter  (as usual $\gm:=0$ is interpreted as the $log$ investor).

The investor receives the random
endowment of the form $B(W_T)$ at time $T>0$, where $B$ is a
non-constant,  bounded and continuous function.
The payoff $\vp$ whose $B$-conditional Davis prices we are computing, as well
as the quantities $\underline{\vp}$ and $\overline{\vp}$ are
defined exactly as in Example 1 above.
We also define the following quantities
$$
\underline{B}(\eps) := \inf_{a\in\R}
\Big(B(a)+\eps \varphi(a)\Big),
{\quad\overline{B}(\eps) := \sup_{a\in\R}\Big(\eps \varphi (a)-B(a)
 \Big),}\quad \eps\ge0.
$$

Proposition 2.4 in \cite{LSZ16} states that the dual optimizer
$\hat{\Q}\in$ ba$(\P)$ for the utility-maximization problem with the random
endowment of the form $B(W_T)$ has a
non-trivial singular part in the Yosida-Hewitt decomposition $\hat{\Q} =
\hat{\Q}^r + \hat{\Q}^s$ after a possible shift of the function $B$ by a
constant. Moreover, such a shift can always be arranged so as to keep the
values of $B$ positive and bounded away from $0$. Therefore, we assume,
without loss of generality that such a shift has already been performed, so
that, in particular, we have $\underline{B}(0)>0$.
 This loss-of-mass
property for $\hat{\Q}^r$ can be partially quantified as follows: Theorem 3.7 in
\cite{GLY16} and Proposition 3.2 in \cite{LZ07} allow us to write $\tRN
{\hat{\Q}^r}{\PP} =
\hat{Y}_T$ where
$$
d\hat{Y}_t = -\hat{Y}_t \big( \tfrac\mu\sigma dZ_t + \hat{\nu}_t dW_t\big),  \quad \hat{Y}_0>0,
$$
for some process $\hat{\nu}\in \sL^2$. The presence of the non-trivial
singular part $\hat{\Q}^s$ implies that $\hat{Y}$ is a strict local
martingale, i.e., $\E[\hat{Y}_T]<\hat{Y}_0$.

Example \ref{example3} takes care of the conditions of Theorem \ref{thm:derb}
dealing with unique superreplicability.
Indeed, both $-B$ and $-(B+\eps \vp)$ are of the form treated there,
and are, therefore, uniquely superreplicable by $-\underline{B}$ and
$-\underline{B}(\eps)$, respectively, for $\eps\ge0$.

The last step before Theorem \ref{thm:derb} is applied is to
simplify the two $\eps$-limits appearing in \eqref{equ:der-2}. That is an
easy task thanks to the fact that the random variable
$\oo{\eps} \big( \underline{B+\eps \vp}-\underline{B}\big)$ is constant and
equal to $\oo{\eps}\big( \underline{B}(\eps) - \underline{B}(0) \big)$.
Theorem \ref{thm:derb} guarantees that this quotient admits a left and a
right limit at $\eps=0$ and we introduce the following notation
\[ \underline{B}'(0+) := \lim_{\eps \downto 0}
\oo{\eps}\big( \underline{B}(\eps) - \underline{B}(0) \big) \eand
 \underline{B}'(0-) := \lim_{\eps \upto 0}
\oo{\eps}\big( \underline{B}(\eps) - \underline{B}(0) \big).\]
The total mass in $\hat{\QQ}^s$ is given by $\hat{Y}_0-\EE[ \hat{Y}_T]$,
and, so the interval of
$B(W_T)$-conditional Davis prices for the payoff $\varphi(W_T)$ is given by
$[\underline{p},\overline{p}]$ where
\begin{align*}
\underline{p} &:= \tfrac1{\hat{Y}_0}\E\left[\hat{Y}_T\big(\varphi(W_T)-\underline{B}'(0+)\big)\right]+
\underline{B}'(0+),\\ \overline{p}&:=\overline{B}'(0+)
- \tfrac1{\hat{Y}_0}\E\left[\hat{Y}_T\big(\overline{B}'(0+)-\varphi(W_T)\big)\right].
\end{align*}

\subsubsection{Linear Approximation}
We close the paper with a result which complements the pricing
formula of the previous two examples with some asymptotic hedging information.
More precisely, we provide two first-order approximations
to the primal utility maximizer in the Brownian setting used above.
We focus on  the right limit ($\eps \downto 0$), as one gets
the left-limit corrector by applying the result to $-\vp$. As a
preparation, we note that the function $\underline{B}$ is concave and that
its right derivative $\underline{B}'(0+)$ at $0$  satisfies
\[ |\underline{B}'(0+)| \le \sup_{a\in\R}|\vp(a)|, \]
and remind the reader that both $B$ and $\vp$ are normalized so that
$\underline{B}(0)>0$. As always, $\hpi$ denotes the primal optimizer for
the utility-maximization problem with the random endowment $B=B(W_T)$ and
$\hat{Y}$ is the common regular part of all dual optimizers.

\begin{proposition} \label{cor:picorrector} In the setting described in the
beginning of subsection \ref{sse:two}, and under Assumption
\ref{ass:U-prime}, the process \begin{align*}
(1 + \eps\tfrac{\underline{B}'(0+)}{\underline{B}(0)}) \hat{\pi}
\end{align*}
 is first-order optimal in the sense that
$$
\valu\big(B+\eps \vp(W_T)\big) - \E\left[U\Big( \big(1+\eps\tfrac{\underline{B}'(0+)}{\underline{B}(0)}\big)(\hat{\pi}\cdot S\big)_T +
B+\eps\vp(W_T)\Big)\right]  = o(\eps)
$$
as $\eps\searrow 0$.
\end{proposition}
\begin{proof} We base the proof on the abstract expression
\eqref{equ:der-exp} for the right derivative of the value function of
Proposition \ref{pro:limit}.

We let $\delta^\eps$ be defined by \eqref{deltaeps} in the proof of
Proposition \ref{pro:limit}, which thanks to Example \ref{example3},
takes the simple form
\[ \delta^\eps_t := \frac1{\underline{B}(0)} \frac{\underline{B}(\eps) -
\underline{B}(0)}\eps \hat{\pi}_t.\] Then we have
\begin{align*}
\sup_{\delta \in \Delta(\vp)} \EE[ \hY \dest ] &=
\lim_{\eps\searrow 0}\E\Big[\hat{Y}(\delta^\eps \cdot S)_T\Big]\\
&=\tfrac{\underline{B}'(0+)}{\underline{B}(0)}\E\Big[\hat{Y}(\hat{\pi} \cdot S)_T\Big]
 = \underline{B}'(0+)\hat{\mu}^s(\Omega),
\end{align*}
where the last equality uses \eqref{KL2}.
We can then re-purpose the proof of Proposition \ref{pro:var} to see that
\begin{align*}
& \E\left[U\Big(
\big(1+\eps\tfrac{\underline{B}'(0+)}{\underline{B}(0)}\big)(\hat{\pi}\cdot
S\big)_T + B+\eps\vp\big(W_T\big)\Big)\right] -\valu\big(B\big)
-\eps\Delta= o(\eps),
\end{align*}
as  $\eps\searrow 0$ where we have defined
\begin{align*}
\Delta:& =   
\EE[\hY\vp] +\underline{B}'(0+)\hat{\mu}^s(\Omega).
\end{align*}
It remains to apply the triangle inequality and Proposition \ref{pro:limit}.
\end{proof}

\bibliographystyle{amsalpha}


\end{document}